\documentclass[11pt,letterpaper]{article}
\usepackage[round]{natbib}
\usepackage{fullpage}
\usepackage{amsmath}
\usepackage{amsthm}
\usepackage{amssymb}
\usepackage{mathtools}
\usepackage{comment}
\usepackage{dsfont}
\usepackage[dvipsnames]{xcolor}
\usepackage{xfrac}
\usepackage{enumerate}
\usepackage{color}
\definecolor{DarkBlue}{RGB}{0,0,150}
\usepackage[colorlinks,linkcolor=DarkBlue,citecolor=DarkBlue]{hyperref}
\usepackage{ulem}

\newtheorem{theorem}{Theorem}
\newtheorem{lemma}[theorem]{Lemma}
\newtheorem{definition}{Definition}
\newtheorem{assumption}{Assumption}
\newtheorem{observation}[theorem]{Observation}
\newtheorem{proposition}[theorem]{Proposition}
\newtheorem{corollary}[theorem]{Corollary}

\renewcommand{\b}{\mathbf{b}}

\newcommand{\h}{\mathbf{h}}
\newcommand{\s}{\mathbf{s}}
\newcommand{\z}{\mathbf{z}}
\renewcommand{\t}{\mathbf{t}}

\newcommand{\N}{[n]}

\newcommand{\E}{\mathbb{E}}
\newcommand{\Mask}{\textsf{M}}

\newcommand{\epicir}{{\textsf{EPIC-IR}}}
\newcommand{\epic}{{\textsf{EPIC}}}
\newcommand{\cepicir}{{\textsf{C-EPIC-IR}}}
\newcommand{\cepic}{{\textsf{C-EPIC}}}
\newcommand{\epir}{{\textsf{EPIR}}}
\newcommand{\cepir}{{\textsf{C-EPIR}}}
\newcommand{\epbb}{\textsf{EPBB}}
\newcommand{\GVA}{{\textsf{GVA}}}
\newcommand{\OPT}{{\textsf{OPT}}}
\newcommand{\REV}{{\textsf{REV}}}
\newcommand{\MGVA}{{\textsf{M-GVA}}}

\newcommand{\ycfixed}[1]{}

\begin{document}
	
	\title{Cursed yet Satisfied Agents}
	
	\author{
	%	Submission \#104
		Yiling Chen\thanks{Harvard; {\tt yiling@seas.harvard.edu}}
		\and
		Alon Eden\thanks{Harvard; {\tt aloneden@seas.harvard.edu}}
		\and
		Juntao Wang \thanks{Harvard; {\tt juntaowang@g.harvard.edu}}
	}
	\date{}
	
	\maketitle

\begin{quote}
	{\small\raggedright \noindent {\it ``You win, you lose money, and you
			curse.''}}
	
	\hspace{2cm} {\tiny\raggedleft \noindent {--- Kagel and Levin}}
\end{quote}
	
\begin{abstract}
    In real-life auctions, a widely observed phenomenon is \textit{the winner's curse}---the winner's high bid implies that the winner often overestimates the value of the good for sale, resulting in an incurred negative utility. The seminal work of Eyster and Rabin [Econometrica'05] introduced a behavioral model aimed to explain this observed anomaly. We term agents who display this bias ``cursed agents''. We adopt their model in the interdependent value setting, and aim to devise mechanisms that prevent the agents from obtaining negative utility. We design mechanisms that are \textit{cursed ex-post incentive compatible}, that is, incentivize agents to bid their true signal even though they are cursed, while ensuring that the outcome is \textit{ex-post individually rational} (\epir)---the price the agents pay is no more than the agents' \textit{true} value. 
	
	Since the agents might over-estimate the value of the allocated good, such mechanisms might require the seller to make positive (monetary) transfers to the agents in order to prevent agents from over-paying for the good. While the revenue of the seller not requiring $\epir$ might increase when agents are cursed, when imposing $\epir$, cursed agents will always pay less than fully rational agents (due to the positive transfers the seller makes). We devise revenue and welfare maximizing mechanisms for cursed agents. For revenue maximization, we give the optimal deterministic and anonymous mechanism. For welfare maximization, we require {\it ex-post budget balance} (\epbb), as positive transfers might cause the seller to have negative revenue. 
	%We propose a masking operation that takes any deterministic mechanism, and imposes that the seller would not make positive transfers, enforcing $\epbb$. 
	We propose a masking operation that takes any deterministic mechanism, and masks the allocation whenever the seller requires to make positive transfers. The masking operation ensures that the mechanism is both $\epir$ and $\epbb$. 
	We show that in typical settings, $\epbb$ implies that the mechanism cannot make any positive transfers. Thus, applying the masking operation on the fully efficient mechanism results in a socially optimal $\epbb$ mechanism. This further implies that if the valuation function is the maximum of agents' signals, the optimal $\epbb$ mechanism obtains zero welfare. In contrast, we show that for sum-concave valuations, which include weighted-sum valuations and $\ell_p$-norms, the welfare optimal $\epbb$ mechanism obtains half of the optimal welfare as the number of agents grows large.    
\end{abstract}

%\begin{document}
	
	% Title page for title and abstract only.
	%\begin{titlepage}

	%	\maketitle

%	\end{titlepage}
	
	% Paper body

	\section{Introduction}\label{sec:intro}
	
	Consider the following hypothetical game---Alice and Bob each have a wallet, with an amount of money that is known to them, but not to the other player. They each know that the money in the other wallet is distributed uniformly in the range between $\$0$ and $\$100$, independently of the amount of money in their own wallet. The auctioneer confiscates the two wallets, and runs a second price auction on the two wallets  (the highest bidder wins the two wallets, and pays the bid of the other bidder). Say Alice has $\$30$ in her wallet, how should she bid? A na\"ive strategy Alice could take is to calculate the expected amount of money in Bob's wallet, $\$50$, and add it to her amount, resulting in a bid of $\$80$. However, such a bidding strategy ignores the fact that if Bob invokes the same strategy, then conditioned on Alice winning the two wallets, she has more money in her wallet than Bob, implying that Bob's amount is distributed uniformly between $\$0$ and $\$30$. Thus, if both agents invoke the  na\"ive strategy, Alice's utility conditioned on winning is the expected sum in the two wallets, $\$30+\$15=\$45$ \textit{minus} Bob's expected bid conditioned on him losing $\$15+\$50=\$65$, implying a \textit{negative} expected utility of $-\$20$. 

Of course, a rational agent should not incur a negative utility when playing a game. \citet{klemperer1998auctions} introduced the game presented above, and named it ``the wallet game". This is an interdependent value setting (IDV)~\citep{milgrom1982theory}, where each agent has \textit{private} information, termed \textit{signal} (i.e., the amount of money in their own wallet), and a \textit{public} valuation function that takes into account different bidders' private information (in this case, the sum of signals).  
%The following hypothetical game was introduced by Klemperer~\cite{klemperer1998auctions},\yc{This is just the wallet game, right?} where he 
\citet{klemperer1998auctions} analyzed the symmetric equilibrium of rational agents in the wallet game (and introduced several asymmetric equilibria). However, in practice, the observed behavior of agents in the wallet game much resembles the na\"ive strategy rather than an equilibrium that rational agents end up with~\citep{avery1997second}. This phenomenon was first observed by~\citet{capen1971competitive}, three petroleum engineers who observed that oil companies experienced unexpectedly low rates of returns in oil-lease auctions since these companies ``ignored the informational consequences of winning''. %\AE{I like this story as it adds a little color, but we can move it to the related section if it fits better.} \yc{I like it!}

This behavioral bias, where the winner fails to account for the implications of outbidding other agents is commonly referred to as \textit{the winner's curse}, and was consistently observed across many scenarios such as selling mineral rights~\citep{capen1971competitive,lohrenz1983bonus}, book publication rights~\citep{dessauer1982book}, baseball's free agency market~\citep{cassing1980implications, blecherman1996there}, and many others (for more information about empirical evidence for the winner's curse, see Chapter 1 in~\cite{kagel2009common}). As standard game theory cannot account for the observed behavior, \citet{eyster2005cursed} introduced a behavioral model that formalizes this discrepancy. They termed this model as ``cursed equilibrium".

%The effect 

%To explain this observed irrational behavior, \citet{eyster2005cursed} introduced a behavioral model that formalizes this discrepancy. 
%They termed this model as ``cursed equilibrium". 
In their model, agents correctly predict other agents' strategies, but fail to estimate that other agents' actions are correlated with their actual signals (or information), similar to the na\"ive strategy presented above. The extent of this degree of misestimation of the correlation of actions and signals is captured by a parameter $\chi$, where the perceived utility of agents  %termed \textit{cursed}-utility\jt{we did not use the term in our work}, 
is $\chi$ times the expected utility of the agents if the actions and signals of other agents are uncorrelated \textit{plus} $(1-\chi)$ times the actual, correct expected utility, correlating the signals with the actions (see Section~\ref{sec:cursed-eq} for a formal definition). %\yc{Do we want to emphasize that this is expected utility?} 
Having a single parameter to explain the behavioral model of the agents proved to be a very tractable modeling, as this parameter can be easily fitted using real world data to estimate $\chi$, and better predict players behaviors~\citep{eyster2005cursed}.

Existing literature of interdependent values either
\begin{enumerate}[(i)]
	\item analyzes fully rational players' behavior as they `shade' their bid to account for the potential over-estimation of the value% when winning the auction
	~\citep{klemperer1998auctions,milgrom1982theory,wilson1969communications,nobel2021considerations};
	\item studies biased agents' behavior when deploying known mechanisms, where the equilibrium often times implies a negative utility for the bidders (and higher revenue for the seller)~\citep{eyster2005cursed,kagel1986winner,avery1997second,holt1994loser}; or
	\item  exploits the biased behavior of the agents to achieve higher revenue~\citep{bergemann2020countering}.
\end{enumerate}
%Our work does not assume agents are fully rational, 
Our work builds upon the \citet{eyster2005cursed} ``cursed equilibrium" model and views cases where agents experience \textit{actual} negative utility as undesirable; thus, tries to avoid such scenarios. 

One might wonder why a seller might care about a cost incurred by the buyer due to her own bias, especially when the outcome might increase the seller's revenue. However, we note  that this can be highly undesirable due to various reasons:
\begin{itemize}
	\item In many real-life scenarios, such as leasing spectrum bands, violating ex-post IR can be detrimental to society at large.  Perhaps a mobile company overbids on spectrum, winds up bankrupt, and then the public cannot enjoy any cellular service associated with that spectrum lease~\citep{zheng2001high}.  
	\item Companies experiencing revenue loss might feel reluctant to join future  auctions~\citep{hendricks1988empirical}. This may have the adverse affect of reducing the long-term revenue of the seller, and the long-term social welfare in the market.
\end{itemize}
%\that there are real-life scenarios 
%Indeed, 

In this paper, we design mechanisms that are incentive compatible (IC) for \textit{cursed} agents---agents maximize their \textit{biased} utility by reporting their true private information, thus generate a predictable behavior. In order to avoid the winner's curse, the mechanisms we introduce are \textit{ex-post individually rational}, meaning an agent will never pay more than their \textit{true} value. We study the quintessential objectives of revenue and welfare maximization in auction settings with interdependent values.

	%\section{Our Results and Techniques} \label{sec:results_and_techniques}
	%Our model section is quite elaborated. It introduces the involved behavioral model of \cite{eyster2005cursed}, and extends it to ex-post solution concepts. Moreover, our starting point is the extension of Roughgarden and Talgam-Cohen of Myerson's auction theory to interdependent settings. As these lengthy preliminaries are required in order for the reader to fully follow our technical work, there is no hope in fitting much of our technical contribution in the 10 page limit of the submission. For these reasons, we now present a lengthy results and techniques section intended to give the reviewer a clear yet informal statement of our results (Section~\ref{sec:results}), followed by an informal discussion of the techniques used to obtain these results (Section~\ref{sec:techniques}).

	\subsection{Our Results}\label{sec:results}
	
We focus on deterministic and anonymous mechanisms, as they are optimal for interdependent values of fully rational agents~\citep{ausubel1999generalized,maskin1992auctions,RoughgardenT16}. 
We extend the cursed-equilibrium model of \citet{eyster2005cursed} to support a strong truthfulness notion of Cursed Ex-Post Incentive-Compatible (\cepic), the equivalent of ex-post IC in the case of fully rational agents, which is the strongest incentive notion possible for this setting\footnote{The ex-post IC notion is stronger than Bayesian IC and weaker than Dominant strategy IC. In interdependent settings it is impossible to design dominant strategy IC mechanisms while obtaining good performance guarantees.} (see Section~\ref{sec:cursed-eq}). %\yc{I suppose this should refer to the model section. In addition, it's a good idea to mention in the model section that C-EPIC is the strongest incentive notion possible for our setting.} 
For interdependent values, a deterministic $\cepic$ mechanism corresponds to a threshold allocation rule, which takes as input other bidders' reports, and returns the minimum bid from which an agent starts winning.

After establishing the incentive notion we are studying in this paper, we turn to take a closer look at the implications of ensuring the mechanisms satisfy ex-post individual rationality (Section~\ref{sec:epir-implication}). 
Our solution concept gives rise to an analogue of the payment identity of \epic\ mechanisms (Proposition~\ref{thm:CEPICIR}). As opposed to the fully rational setting, we might need to set the constant term in the payment identity ($p_i(0,\s_{-i})$ in Equation~\eqref{eq:payment identity}) to be smaller than zero, as the mechanism might make positive transfers to compensate for the over-estimation of values due to the winner's curse.
For a fixed deterministic mechanism, we show how to optimally set this compensation term in a way that maximizes the revenue for the allocation rule, while keeping $\epir$. This has the following implication---fixing the allocation rule, and letting $\chi$ grow decreases the revenue. An interesting conclusion is the following (see Propositions~\ref{prop:rev-mon} and \ref{prop:welfare-mon}):

\vspace{3mm}

\noindent\textbf{Proposition (Revenue and welfare monotonicity).} As the cursedness-parameter $\chi$ increases: (i) the revenue of the revenue-optimal \epicir\ and \epir\ mechanism decreases. (ii) the welfare of the welfare-optimal \epicir, \epir\ and ex-post budget-balanced mechanism decreases.

\vspace{3mm}

%that the revenue and welfare of the revenue-optimal and welfare-optimal mechanisms that are \cepicir\ and \epir\ (and ex-post budget-balance for welfare) decreases as $\chi$ increase (see Propositions~\ref{prop:rev-mon} and \ref{prop:welfare-mon}). 

This is in stark contrast to the case where the mechanism does not require \epir, where cursed-agents are shown to generate more revenue in second price auctions, since the mechanisms take advantage of the possibility of agents paying more than their value~\citep{eyster2005cursed}.  

%This implies that when 

%This is in stark contrast to the results~\cite{eyster2005cursed}, where cursed-agents are shown to generate more revenue in second price auctions, since the mechanisms takes advantage of the possibility of agents paying more than their value.
%\begin{enumerate}
%\item when receiving an item, if the agent's value at the threshold is smaller than the cursed value, then the agent pays her real value at the threshold (the signal from which the agent starts winning), and otherwise,  the agent pays the cursed value at the threshold. 
%\item  
%\item The revenue and welfare of the revenue-optimal and welfare optimal mechanisms that are \cepicir\ and \epir\ (and ex-post budget-balance for welfare) decreases as $\chi$ increase.
%\end{enumerate}

Building upon our understanding of combining individual-rationality constraints with incentive-compatibility for agents who suffer from the winner's curse, we turn to revenue maximization (see Section~\ref{sec:revenue}). We show that designing a revenue optimal $\cepicir$ and $\epir$ deterministic mechanism decomposes into a separate problem for each agent $i$ and other agents' signals $\s_{-i}$. %\AE{We should verify this is clear enough in the section itself.}
That is, the mechanism designer's task reduces to find an optimal threshold for winning the auction for agent $i$  for every set of signals $\s_{-i}$ other agents might declare. We show how to optimally set such a threshold, resulting in a revenue-optimal mechanism (see Theorem~\ref{thm:rev-optimal}). %The optimal threshold rule is similar in some aspects to 
%\red{bears certain intuition}, and 
\vspace{3mm}

\noindent\textbf{Theorem (Revenue optimal mechanism).} The revenue-optimal mechanism is a threshold mechanism, and the threshold rule is given via Theorem~\ref{thm:rev-optimal}.

\vspace{3mm}

We discuss interesting similarities and differences of the optimal threshold rule from the case where the agents are fully rational, and the case where the seller does not require the mechanism to keep $\epir$ constraints at the end of Section~\ref{sec:revenue}.

Social welfare maximization is a more nuanced task, as just aiming for $\cepicir$ and $\epir$ might result in the auctioneer needing to pay all the agents (even the losers). In fact, we show a scenario that the welfare-optimal $\cepicir$  and $\epir$ mechanism incurs a huge revenue loss of $\Omega(n\log n)$, where the expected optimal welfare is $O(n)$. This holds even for a very simple setting where $i$'s valuation is $v_i(\s)=s_i+\frac{1}{2}\sum_{j\neq i}s_j$ and signals are sampled independently from the $U[0,1]$ distribution (see Proposition~\ref{prop:negative-rev}). 
%the auctioneer losing a huge amount of money when maximizing, due to the positive transfers the auctioneer needs to make in order to obtain $\epir$ (see Proposition~\ref{prop:negative-rev}). 
Therefore, when maximizing welfare, we add a requirement of ex-post budget-balance (\epbb); that is, the seller never has negative revenue.

A trivial way to ensure $\epbb$ is to set the threshold rule such that it never sells whenever selling implies the seller will need to make positive transfers. We present a masking operation that does exactly this. 
%Given a threshold rule, the masking of the threshold rule only allocates whenever the original threshold rule allocates, and the allocation does not require the seller to make positive transfers. 
Given a threshold rule, the masking of the threshold rule allocates in the maximal subset of cases where the original threshold rule made an allocation \textit{and} the allocating induces no positive transfer from the seller.
One might wonder whether one can design a mechanism that is $\epbb$, while still selling at scenarios that require the mechanism to make positive transfers to some bidders, increasing the expected welfare of the mechanism. We answer this question in the negative (see Theorem~\ref{thm:no-positive-transfers}).

\vspace{3mm}

\noindent\textbf{Theorem ($\epbb$ $\equiv$ no positive transfers).} Under natural conditions, a mechanism is $\epbb$ if and only if the mechanism \textit{never} makes positive transfers.

\vspace{3mm}

%---we show that a mechanism is $\epbb$ if and only if the mechanism \textit{never} makes positive transfers to any bidder (see Theorem~\ref{thm:no-positive-transfers}).
This characterization drives the following implication (Proposition~\ref{prop:gva-optimal}).

%yields the following that the welfare-optimal $\epbb$ is a masking of the socially optimal mechanism that is not $\epbb$ (see Proposition~\ref{prop:gva-optimal}).
\vspace{3mm}

\noindent\textbf{Proposition (Welfare optimal mechanism).} Under natural conditions, the welfare-optimal $\epbb$ mechanism is a result of masking the welfare-optimal allocation.

\vspace{3mm} 

Applying the above proposition to the wallet game example in the introduction implies that in that example, the seller should allocate the item only when the loser has at least $\$50$ in her wallet, since this is the case where the winner does not overestimate the amount of money in the two wallets.
%The proposition above applies to the wallet game example in the introduction. 

%\blue{For the wallet game where the value of the item is the sum of the signals of all agents, the welfare-optimal $\epbb$ mechanism is to allocate the item to the winner only when the losers also have high-end signals for the item.}

We notice that for some valuation functions, as the maximum of agents' signals, this implies an $\epbb$ mechanism can never sell the item, resulting in zero welfare. On the positive side, we introduce a new family of valuations, \textit{Concave-Sum valuations}, which generalizes well studied classes of valuations such as weighted-sum valuations\footnote{Weighted-sum valuations take the form $v_i(\s)=s_i + \beta\sum_{j\neq i} s_j$ for $\beta \in (0,1]$. Note that the valuations in the wallet game example are a special case of weighted-sum valuations.} (e.g., the wallet game) and $\ell_p$-norms of signals for a finite $p$. We show that for this class of valuations, the optimal $\epbb$ mechanism approximates the optimal welfare (Theorem~\ref{thm_masked_SW}). 

\vspace{3mm}

\noindent\textbf{Theorem (Welfare approximation).} For concave-sum valuations, the welfare-optimal $\epbb$ mechanism gets at least half of the fully efficient allocation as the number of agents grows large.
%Under natural conditions, the welfare-optimal $\epbb$ mechanism is a result of masking the welfare-optimal allocation.

\vspace{3mm}

% demonstrate that for Concave-Sum valuations, a family of valuation functions that includes weighted-sum valuations\footnote{Weighted-sum valuations take the form $v_i(\s)=s_i + \beta\sum_{j\neq i} s_j$ for $\beta \in (0,1]$. Note that the valuations in the wallet game example are a special case of weighted-sum valuations.} and $\ell_p$-norms of signals for a finite $p$, the welfare-optimal $\epbb$ mechanism gets at least half of the fully efficient allocation as the number of agents grows large (see Theorem~\ref{thm_masked_SW}).

	%\subsection{Our Techniques}\label{sec:techniques}
	%\input{techniques}

\section{Related work}
Our work investigates the problem of auction design for the agents who suffer from the winner's curse. 
%The winner's curse has been regarded as a prevalent evidence that ``some subjects  frequently fail to play equilibrium strategies.''~\cite{kagel1986winner}
As a behavioral anomaly~\citep{thaler1992winner}, the winner's curse has been documented and analyzed by a large literature via both field studies and lab experiments. 
Field evidence of the presence of the winner's curse has been discovered in auction practices across a wide array of industries, which range from the book industry~\citep{mcafee1987auctions} and the market of baseball players~\citep{cassing1980implications} to the offshore oil-drilling leases~\citep{capen1971competitive,hendricks1988empirical,hendricks1987information,porter1995role}. Also, a large amount of lab experimental results~\citep{bazerman1983won,kagel1986winner,avery1997second,charness2009origin,ivanov2010can,bernheim2019handbook} show that the winner's curse occurs under various conditions, which differ in multiple dimensions, such as auction format (e.g., first-price, second-price, Dutch auction and English auction), number of participants, valuation functions and signal structures. \citet{kagel2002bidding} provided a comprehensive review for such experimental studies. Furthermore, most of these studies, such as observational ones~\citep{hendricks1988empirical,hendricks1987information} and lab experimental ones~\citep{bazerman1983won,kagel1986winner,avery1997second,charness2009origin}, demonstrate that the bidders who suffered from the winner's curse not only experienced a reduce in the profit than anticipated, but could also be worse off upon winning, i.e., receive a negative net profit. %In this paper, one of our goal is to design auction mechanisms to avoid inducing negative profit on bidders. 

As discussed in the introduction, we adopt the cursed equilibrium model introduced by~\citet{eyster2005cursed} to model the bidding behaviors of agents who suffer the winner's curse. This model suggests that  agents fail to fully appreciate the contingency between other bidders' bids and the auction item's value. This cause is supported by several experimental findings~\citep{bazerman1983won,kagel1986winner,avery1997second,charness2009origin}, and this model has been applied, generalized or analogized to different applications, including analyzing market equilibrium~\citep{eyster2019financial,eyster2011correlation}, designing financial assets~\citep{kondor2015cursed,ellis2017correlation}, unifying theoretical behavioral models~\citep{miettinen2009partially}.

In addition, there is a large literature, including~\cite{wilson1969communications,milgrom1982theory,klemperer1998auctions,bulow2002prices,RoughgardenT16}, studying the interdependent valuation auction, the type of auctions we consider in this paper. Different from our work which designs mechanism for agents who suffer the winner's curse, these papers consider the mechanism design for fully rational agents who play (Bayesian) Nash equilibrium strategies. \citet{milgrom1982theory} introduced the interdependent value model and analyzed the revenue of different auction formats when agents have correlated but private value over the item. Their results imply that fully rational agents who implicitly try to avoid the winner's curse bid more conservatively when there is less information (as in a second price auction) revealed in the auction than more information revealed (as in English auction). \citet{bulow2002prices,klemperer1998auctions} showed the anomalies in certain interdependent valuation auctions that the item price may increase in supply and decrease in the number of bidders and that the item price is sensitive to even a small asymmetry of bidders. They interpreted the anomalies in terms of fully rational bidders taking the winner's curse into consideration. \citet{RoughgardenT16} developed tools to build ex-post incentive compatible mechanisms for general interdependent valuation auctions with fully rational agents. We use their tools to build our mechanisms. 

In contrast to these works, \citet{bergemann2020countering} studied the auction design problem for agents suffering the winner's curse. Their work is the most similar to ours, but with several differences. The major difference is that they aim to achieve interim incentive guarantees, while we aim to achieve stronger ex-post incentive guarantees. As a result of sacrificing the ex-post incentive guarantees, their allocation rule needs not to be monotone, i.e., they can allocate the item to an agent with a lower signal, incurring the winner's blessing instead of the winner's curse. In contrast, when imposing ex-post incentive guarantees, we show that the allocation rule must be monotone. They also consider a common value auction, specifically, only a single function which is the maximum of signals; Therefore, their non-monotone allocation rule will not decrease revenue and social welfare. However, this is not true in the more general valuation settings that we consider. Moreover, they consider agents who are fully cursed, while we consider agents who can be partially cursed. Finally, we study a much more general setting capturing a general family of valuation functions, while they consider a single valuation function. %, defined as the maximum of all agents' private signals, while we consider a more general family of valuation functions. 
%they aim at achieving interim incentive guarantees, while we aim at achieving stronger ex-post incentive guarantees; they consider agents who are fully cursed, while we consider agents who can be partially cursed; they consider a single valuation function for the item, defined as the maximum of all agents' private signals, while we consider a more general family of valuation functions. 

The design of mechanisms when considering agents who act according to a behavioral bias, and not their objective utility, have recently gained traction. Recent examples of this line of research are  finding a market equilibrium for agents who suffer from the endowment effect~\citep{EzraFF20,BabaioffDO18} and designing revenue-maximizing auctions for agents who are uncertainty-averse~\citep{0001GMP18,LiuMP19}, among others. 
%planning for agents experiencing present bias~\cite{KleinbergOR16,KleinbergOR17}, among others.

Finally, approximately optimal mechanisms in the interdependent values model recently gained attention. Works in this domain include simple and approximately optimal revenue maximizing auctions~\cite{RoughgardenT16,li2017approximation,ChawlaFK14} and assumption-minimal welfare maximizing auctions~\cite{EdenFFG18,EdenFFGK19,EdenFTZ21,AmerT21,EdenGZ22,GkatzelisPPS21}

	\section{Model}\label{sec:model}
	
	%\red{We consider the IDV that are commonly studied for rational agents and consider the set of mechanisms know to be optimal. To compare to other interdependent setting}

    We consider \textit{interdependent} valuation (IDV)  settings commonly studied for rational agents. We consider a seller that sells a single indivisible item to a set of $n$ agents with \textit{interdependent} valuations. 
    %single-item multiple-bidder interdependent valuation auction setting. In this setting, there is one indivisible item for sale and a set agents (bidders), denoted as $\N=(1,...,n)$, bidding for the item. 
    Each agent $i$ has a signal $s_i$ as private information. The agents' signals are drawn from a joint distribution $F$ with density $f$ over the support $S_i^n$, where $S_i=[0, \bar{s}_i]$.
    %The agents have \emph{privately-known} signals $s_1,...,s_n$ drawn from a joint distribution $F$ with density $f$ over the support $S_i^n$, where $S_i=[0, \bar{s}_i]$. 
    We use $\s$ to denote the agents' signal profile $s_1,...,s_n$ and use $\s_{-i}$ to denote the signal profile of all agents except agent $i$. We impose the standard assumptions that $f$ is continuous and nowhere zero on the signal space. Each agent $i$ also has a {\it publicly-known} valuation function $v_i:[0,\bar{s}_i]^n\mapsto \mathbb{R}$, which represents the value received by agent $i$ upon winning the auction as a function of all bidders' signals.
    %the joint signal information obtained by all agents collectively. 
    We adopt the following standard assumptions in the valuation function $v_i(\cdot)$:
	\begin{itemize}
	    \item[--] Non-negative and normalized, i.e., $\forall \s, v_i(\s)\ge 0$ and $v_i(\mathbf{0})=0$.
	    \item[--] Continuously differentiable. 
	    \item[--] Monotone-non-decreasing in all signals and monotone-increasing in agent $i$'s signal $s_i$.
%	    \item[--] Finite expectation $\E_\s[v_i(\s)]<+\infty$. \jt{Do we need this?}
	\end{itemize}
	%\AE{I removed the mention of virtual values as it is mentioned in the related works section.}
	%We define the \emph{conditional virtual value} of the item to  agent $i$ as $\phi_i(\s)=v_i(\s)-\frac{\partial v_i(\s)}{\partial s_i}\cdot\frac{1-F(s_i|\s_{-i})}{f(s_i|\s_{-i})}$. 
	%$\phi_i(\s)=v_i(\s)-v_i'(\s)\cdot\frac{1-F(s_i|\s_{-i})}{f(s_i|\s_{-i})}$, where $v_i'(\s)$ refers to the partial derivative of $v_i$ to $s_i$, i.e., $v_i'(\s)=\frac{\partial v_i(\s)}{\partial s_i}$ 
	
	We consider the following properties, introduced in the seminal paper of \citet{milgrom1982theory}:
	\begin{enumerate}
	    \item[--] (Signal symmetry) $S_1=...=S_n=S$, where $S=[0, \bar{s}]$ for some $\bar{s}$, and $f(\s)=f(\t)$ for any signal profile $\s$ and its arbitrary permutation $\t$. 
	    \item[--] (Valuation  symmetry) For any $i,j$,  $v_i(s_i, \s_{-i})=v_j(t_j, \t_{-j})$ as long as $s_i=t_j$ and $\t_{-j}$ is a permutation of $\s_{-i}$. 
	    
	    \item[--] 
	    (Signal affiliation) For any pair of signal profiles $\s$ and $\s'$, it always holds that   $f(\s\vee\s')f(\s\wedge\s')\ge f(\s)f(\s')$, where
	    $(\s\vee\s')$ is the component-wise maximum, and $(\s\wedge\s')$ is the component-wise minimum.\footnote{Affiliation is a common form of positive correlation, which generalizes the common case of independent distributions. }
	\end{enumerate}
These conditions are standard, and were considered by many papers in the literature, including \cite{eyster2005cursed,RoughgardenT16}. For welfare maximization, we focus on a special form of affiliation, and consider a case where all signals are sampled i.i.d.

%We adopt the standard signal affiliation when considering revenue maximization, while focus on a special form of signal affiliation--all signals are i.i.d., when considering the social welfare maximization.

	%Some of our results require the single crossing 
	%\jt{Add single-cross condition for symmetric agents: agent with higher signal have higher valuations.} %\AE{I think single crossing follows from value symmetry and monotonicity in signal.}
	
	By the revelation principle, we consider without loss of generality
	%\footnote{The only exception is the English auction, which is not strategically equivalent to   See an argument by ~\citet{RoughgardenT16}.}
	direct mechanisms, in which agents directly report their private signals $\s$ and then the auctioneer determines the auction outcome according to a pre-announced mechanism $M=\{(x_i,p_i)\}_{i\in\N}$. Here, $x_i:[0,\bar{s}_i]^n\mapsto [0,1]$ is the  allocation rule, specifying agent $i$'s winning probability, and $p_i:[0,\bar{s}_i]^n\mapsto \mathbb{R}$ is the payment rule, specifying agent $i$'s payment. 
	
	We study deterministic and anonymous mechanism, as these are optimal for rational agents in IDV settings~\citep{RoughgardenT16,ausubel1999generalized,maskin1992auctions}. This allows an easy comparison to existing results in the literature.  
 
	A mechanism is {\it deterministic} if $x_i(\s)\in\{0,1\}$ for any $i$ and $\s$. A mechanism is {\it anonymous} if for any $\s$ and any permutation $\t$ of $\s$, $x_i(\s)=x_j(\t), p_i(\s)=p_j(\t)$ whenever $s_i=t_j$.

	We make two technical assumptions  about allocation rule $x_i$ for simplicity of exposition: 
	First, we do not allocate the item to an agent who reports a zero signal. 
	
	\begin{assumption} \label{assump:zero-no-alloc} 
	$x_i(0,\s_{-i})=0$ for every $i$ and $\s_{-i}$.
	\end{assumption}
	
	Second, we do not allocate when there is a tie in the highest reported signals. 
	
	\begin{assumption} \label{assump:tie-no-alloc}
	$x_i(\s)=0$ for every $i$ whenever $|\arg\max_i \{s_i\}|>1$. 
	\end{assumption}
	
	Since $f$ is continuous, the events in both assumptions have zero probability measure, and therefore, can be ignored without affecting the expected social welfare or revenue of the mechanism. Moreover, these assumptions are without loss for the mechanisms we consider as implied by Lemma~\ref{lem:max-alloc} in Section~\ref{sec_det_cepicir}.

% 	We assume the following two assumptions. The first assumption is without loss of generality, and the second assumption is without loss of generality for deterministic and anonymous mechanisms \cepicir\ mechanisms.
	
% 	\noindent\textbf{Assumption 1.} $x_i(0,\s_{-i})=0$ for every $i$ and $\s_{-i}$.
	
% 	\noindent\textbf{Assumption 2.} $x_i(\s)=0$ for every $i$ whenever $|\arg\max_i \{v_i(\s)\}|>1$.
    
%     We first notice that since $f$ is continuous, the events in Assumptions 1 and 2 have zero measure, and therefore, can be ignored without affecting the welfare or revenue of the mechanism. Moreover, given a monotone allocation rule $x$, setting $x_i(0,\s_{-i}$ preserves monotonicity.
%     Given a deterministic monotone allocation rule $x$, setting $x_i(\s)=0$ whenever $|\arg\max_i \{v_i(\s)\}|>1$ preserves monotonicity for deterministic, anonymous $\cepicir$ mechanisms according to Lemma~\ref{lem:max-alloc}.

	We use $\b=(b_1,...,b_n)$ to denote the reported signal profile (bid profile) of agents. Agents have quasilinear utilities---the utility of each agent $i$ given private signal profile $\s$ and bid profile $\b$ under mechanism $M$ is 
	$$u_i(\b,\s)=x_i(\b)v_i(\s)-p_i(\b).$$
	
		\subsection{The Winner's Curse---A Behavioral Model}% \& Main Desirable Incentive Properties}
	\label{sec:cursed-eq}

% 	To illustrate, let $\sigma(\b|\s)$ denote the probability of bidding $\b$ given signal profile $\s$ under certain bidding strategy $\sigma$. A fully rational (Bayesian) agent selects a bid maximizing her expected utility \begin{align*}
% 	\E_{\b_{-i},\s_{-i}}[u_i(\b,\s)] 
% 	=  & 
% 	\int_{\s_{-i}\in S^{n-1}} f(\s_{-i}|s_i)\cdot\left(\int_{\b_{-i}\in S^{n-1}}\sigma(\b_{-i}\vert\s_{-i})u_i(\b,\s)d\b_{-i}\right)d\s_{-i}\\
% 	= &  
% 	\int_{\s_{-i}\in S^{n-1}} f(\s_{-i}|s_i)\cdot\left(\int_{\b_{-i}\in S^{n-1}}\sigma(\b_{-i}\vert\s_{-i})(v_i(\b,\s)-p_i(\b))d\b_{-i}\right)d\s_{-i}
% 	\end{align*}

% 	\begin{align*}
% 	    & \int_{\s_{-i}\in S^{n-1}} f(\s_{-i}|s_i)\cdot\left(\int_{\b_{-i}\in S^{n-1}}\sigma(\b_{-i}\vert\s_{-i})u_i(\b,\s)d\b_{-i}\right)d\s_{-i}\\
% 	    = &  \int_{\s_{-i}\in S^{n-1}} f(\s_{-i}|s_i)\cdot\left(\int_{\b_{-i}\in S^{n-1}}\sigma(\b_{-i}\vert\s_{-i})(v_i(\b,\s)-p_i(\b))d\b_{-i}\right)d\s_{-i}
% 	\end{align*}
	
% 	For example, in a two-player second price auction with an interdependent valuation 

	We adopt the widely studied behavioral model, namely the {\it cursed equilibrium model}, introduced by~\citet{eyster2005cursed} to explain the occurrence of the winner's curse. In this model, agents fail to incorporate the contingency between the other bidders' actions and their signals, which determine the value of the auctioned item, but succeed in reasoning other parts of the game. To illustrate, let $\sigma$ denote a bidding strategy profile of agents and $f_\sigma$ denote the probability density of bids and signals under strategy profile $\sigma$, e.g., $f_{\sigma}(\b_{-i}|\s_{-i})$ represents the probability density of other bidders bidding $\b_{-i}$ when having signals $\s_{-i}$.
	Given the strategy profile $\sigma$, a fully rational %(Bayesian) 
	agent with signal $s_i$ estimates the probability density of other agents receiving $\s_{-i}$ and bidding $\b_{-i}$ as $$f_\sigma(\b_{-i},\s_{-i}|s_i) = f(\s_{-i}|s_i)f_\sigma(\b_{-i}|\s_{-i}).$$ Consequently, suppose the other agents follow strategy $\sigma_{-i}$, such an agent estimates her expected utility when having signal $s_i$ and bidding $b_i$ as follows:
	\begin{align*}
	EU_i(b_i,s_i;\sigma_{-i})
	=  
	& 
	\int_{\s_{-i}\in S^{n-1}}\int_{\b_{-i}\in S^{n-1}} f(\s_{-i}|s_i)\cdot f_\sigma(\b_{-i}\vert\s_{-i})\big(x_i(\b)v_i(\s)-p_i(\b)\big)d\b_{-i}d\s_{-i}
	\end{align*}
	In contrast, an agent who fully neglects the contingency between other agents' actions and their signals estimates, as the na\"ive agent in the wallet game example, the counterpart probability density as if $\s_{-i}$ and $\b_{-i}$ are independent conditioned on her knowledge $s_i$:
	$$\tilde{f}_\sigma(\b_{-i},\s_{-i}|s_i)=f(\s_{-i}|s_i)f_\sigma(\b_{-i}|s_i),$$
	where $f_\sigma(\b_{-i}|s_i)=\int_{\s_{-i}}f(\s_{-i}|s_{i})f_\sigma(\b_{-i}|\s_{-i})d\s_{-i}.$\footnote{\citet{eyster2005cursed} suggested that the agents succeed in reasoning or perceiving all other parts of the game, except the contingency between other agents' signals and actions. Therefore, agents get the correct $f_\sigma(\b_{-i}|s_i)$, a key assumption made in Eyster and Rabin's behavioral model.} \citet{eyster2005cursed} further introduce a cursedness parameter $\chi$ to model the case where an agent partially neglects this contingency such that she considers the counterpart probability density as 
	$${f}_\sigma^\chi(\b_{-i},\s_{-i}|s_i)=(1-\chi)f_\sigma(\b_{-i},\s_{-i}|s_i)+\chi\tilde{f_\sigma}(\b_{-i},\s_{-i}|s_i).$$ 
	
	An agent with $\chi=0$ is a fully rational agent, and as we will see later, an agent with $\chi>0$ is possible to experience the winner's curse. We refer to such an agent with $\chi>0$ as a \emph{cursed agent} for short, and to an agent with $\chi=1$ as a \emph{fully cursed agent}. By analogy with a fully rational agent estimating her expected utility, an agent with parameter $\chi$ (falsely) estimates her expected utility given $\sigma_{-i}$ as:
	\begin{equation}EU^\chi_i(b_i, s_i;\sigma_{-i})=
	\int_{\s_{-i}\in S^{n-1}}\int_{\b_{-i}\in S^{n-1}} {f_\sigma}^\chi(\b_{-i},\s_{-i}|s_i) \big(x_i(\b)v_i(\s)-p_i(\b)\big)d\b_{-i}d\s_{-i}
	\label{eq_raw_EU}
	\end{equation}
	
	To explain and predict the winner's curse phenomenon, \citet{eyster2005cursed} suggested that agents generally play the equilibrium strategy with respect to this misperceived utility $EU^\chi_i(b_i,s_i;\sigma)$ for some parameter $\chi>0$, instead of $EU_i(b_i,s_i;\sigma)$. They referred to this equilibrium as the {\it $\chi$-cursed equilibrium}. Formally, a strategy profile $\sigma$ forms a $\chi-$cursed equilibrium if it holds that for every agent $i$, $$\sigma(b_i|s_i)>0 \iff 
	b_i\in\text{arg}\max  EU_i^\chi(b_i, s_i;\sigma).$$
	In the above definition, $\chi=0$ gives the definition of the classic {\it Bayes-Nash equilibrium} (BNE). The Na\"ive strategy of the aforementioned wallet game is an example of a cursed equilibrium of fully cursed agents ($\chi=1$). 
	We refer the reader to Appendix~\ref{sec:curse-example} for an illustrative example of how a $\chi$-cursed equilibrium leads to a winner's curse in the wallet game.

	Next, we illustrate how $\chi$-cursed equilibrium relates to the winner's curse. Note that Eq.~(\ref{eq_raw_EU}) can be rewritten\footnote{This rewriting result is given by~\citet{eyster2005cursed}. We present a derivation in the Appendix for completeness.} as follows:
	\begin{equation}EU_i^\chi(b_i,s_i;\sigma_{-i})
	=  
	\int_{\s_{-i}\in S^{n-1}}\int_{\b_{-i}\in S^{n-1}} f(\s_{-i}|s_i)\cdot f_\sigma(\b_{-i}\vert\s_{-i})\left(x_i(\b)v^\chi_i(\s)-p_i(\b)\right)d\b_{-i}d\s_{-i},
	\label{eq_rewrite_EU}
	\end{equation}
	where 
	\begin{equation}
	\label{eq_cursed_valution}
	v_i^\chi(\s):=(1-\chi)v_i(\s)+\chi\E_{\tilde{\s}_{-i}}[v_i(\tilde{\s}_{-i},s_i)].\end{equation}
	This rewriting shows that 
	the utility $EU^\chi(b_i,s_i;\sigma_{-i})$ optimized by an agent with valuation function $v_i$ in the $\chi$-cursed equilibrium  is the same as the expected utility optimized by a fully rational agent with valuation function $v_i^\chi$ in a BNE. Thus, we have the following proposition from \cite{eyster2005cursed}. \begin{proposition} [\cite{eyster2005cursed}]
	In the IDV setting, the $\chi$-cursed equilibrium strategy profile of agents with valuation function $v_i$ is the same to the BNE strategy profile of agents with a modified valuation function $v_i^\chi$.
	\label{prp_CE}
	\end{proposition}
	We name the expression $v_i^\chi$ as the {\it cursed valuation function} of $v_i$. It reflects the hypothetical value of the item to the cursed agent. It contains two part---the $(1-\chi)v_i(\s)$ part reflects the part of the item's value which the agent perceives through successful contingent thinking and the $\chi\E_{\tilde{\s}_{-i}}[v_i(\tilde{\s}_{-i},s_i)]$ part reflects the part of the item's value which the agent perceives when she fully ignore the contingency between other bidders' signals and bids. Therefore, a winner $i$ faces the winner's curse whenever $\E_{\tilde{\s}_{-i}}[v_i(\tilde{\s}_{-i},s_i)]>v_i(\s)$; i.e., bidder $i$ might get an item of which the value is less than $i$ anticipated. Moreover, $i$ might suffer a negative utility when the payment, which can be as high as $\E_{\tilde{\s}_{-i}}[v_i(\tilde{\s}_{-i},s_i)]$ for a fully cursed agent, turns out larger than $v_i(\s)$. We refer to the conceptual utilities built upon the cursed valuation functions  $v_i^\chi$ as {\it cursed utilities.}

	\citet{eyster2005cursed} showed with empirical wallet game data that $\chi$ has a 95\% confidence interval of [0.59, 0.67] with 0.63 the optimal fit. Any $\chi>0$ predicts agents' bids better than the BNE strategy.
    
    We make the following key assumption:	
    \begin{assumption}[Seller knows $\chi$]
        We assume the seller knows the value of $\chi$, that is, the seller knows the extent of which the agents exhibit the cursedness bias. 
    \end{assumption}
	The above assumption can be justified by the following: (i) empirical studies discussed above, showing one can accurately estimate the value of $\chi$; (ii) moreover, the seller can observe the practical behavior of the agents, and their profit, in order to adjust the value of $\chi$, and update the devised auction, if the estimated value of $\chi$ seems to be inaccurate.
	If the assumption does not hold, it is still worthy studying the problem under this assumption for following reasons.  First, we show that  misestimating the value of $\chi$ by $\epsilon$ still leads to an %$\epsilon\cdot v_i(\bar{s},...,\bar{s})$-
	approximate $\cepicir$ mechanism (Proposition~\ref{prp_robust}), and thus using $\chi$ with a small estimation error still preserves some degree of incentive compatibility. Second, devising mechanisms when assuming knowing the value of $\chi$ leads to many interesting theoretical findings, which have implications on designing mechanisms for agents who suffer from winner's curse.
	An example of one such implication is that there exists a tension between ensuring that the agents would not experience negative profit due to their inability to reason about their utility and the revenue of the mechanism. In order words, to ensure non-negative utility for agents who may suffer from the winner's curse, the mechanism has to sacrifice some portion of the revenue. %Therefore, in order to increase its revenue, a mechanism designer should exert effort into explaining to the bidders  the contingency between other agents' actions and their information.
	\ycfixed{I didn't understand the last part. Does explaining to the bidders the contingency mean that the mechanism designer will help agents to form the correct beliefs (and becoming rational)? This doesn't seem to be what you want to express. Also, how does this relates to knowing the cursed bias?}

	\ycfixed{The previous sentence or so doesn't read right to me. Was it finished?} \ycfixed{I don't think we need to go into details of the wallet game example here. We can stop here and just explain that bidding according to cursed valuation may lead to bid that is higher than $v_i(\s)$ and hence winner's curse. Then, it's useful to give the empirical findings about the value of $\chi$ briefly.}

    \subsection{Incentive Properties for Cursed Agents and Other Desirable Properties}
    
	Bearing above behavioral implications of  agents with parameter $\chi$ in mind, a natural generalization of the interim IC concept from fully rational agents to agents with parameter $\chi$ is the following. %should be as follows.
	\begin{definition}
	A mechanism $M=\{(x_i,p_i)\}_{i\in\N}$ is {\it interim incentive compatible for agents with parameter $\chi$}, if for all $i,b_i, s_i$, and for the truth-telling strategy $\sigma^*$, it holds that  
	$$EU^\chi(s_i, s_i;\sigma^*_{-i})\ge EU^\chi(b_i, s_i;\sigma^*_{-i}).$$
	\end{definition}
	In other words, a mechanism is interim incentive compatible for agents with parameter $\chi$ if truthful-reporting is a $\chi$-cursed equilibrium. 
	For fully rational agents ($\chi$=0), the above definition coincides with the standard interim IC definition~\citep{RoughgardenT16}, where truthful-reporting forms a BNE (that is, a $0$-cursed equilibrium).\ycfixed{So it coincide with Bayesian Nash Equilibrium for fully rational agents, right? If so, mention BNE.}

	We further extend this idea to obtain a \textit{stronger} IC notion.
	We consider a cursed agent's expected utility when having signal $s_i$, while the bid profile is $\b=(b_i,\b_{-i})$. A fully rational agent will correctly estimate the degree to which other agents' signals $\s_{-i}$ are contingent on their bids $\b_{-i}$, setting this probability as $f_\sigma(\s_{-i}|\b_{-i},s_i)$, while a fully cursed agent will think the true type is independent of agents' bids, estimating this probability as $f(\s_{-i}|s_i)$. Therefore, an agent with parameter $\chi$ will assess her expected utility given her signal $s_i$, bid $b_i$ and others bidding $\b_{-i}$ given strategy $\sigma_{-i}$ as:
    \begin{eqnarray*}
	EU^\chi_i(\b, s_i;\sigma_{-i})
	& = & 
	\int_{\s_{-i}\in S^{n-1}}
	\left(
	(1-\chi)f_\sigma(\s_{-i}|\b_{-i},s_i)u_i(\b,\s)+\chi f(\s_{-i}|s_i)u_i(\b,\s)\right)d\s_{-i}.
	\end{eqnarray*}

    Note that we have the following relationship between $EU_i^\chi(\b,s_i;\sigma_{-i})$ and $EU_i^\chi(b_i,s_i;\sigma_{-i})$:\footnote{We present the derivation of Equation~\eqref{eq_EPEU_ICEU} in Appendix~\ref{sec:bic-epic-der}.}

 	\begin{equation}EU_i^\chi(b_i,s_i;\sigma_{-i})=\int_{\b_{-i}\in S^{n-1}}f_\sigma(\b_{-i}|s_i)EU_i^\chi(\b,s_i;\sigma_{-i})d\b_{-i}.\label{eq_EPEU_ICEU}\end{equation}
 	
 	Therefore, we can naturally define the ex-post incentive properties for agents with parameter $\chi$ as follows. \begin{definition}[Cursed ex-post incentive compatibility and individually rationality (\cepicir)]
 	\label{def_cepicir}
	Given a cursedness parameter $\chi$, a mechanism is cursed ex-post incentive compatible (\cepic) if for every $i$, $s_i$ and $\s_{-i}$, and truthfully-reporting strategy $\sigma^*$,
	$$EU_i^\chi(\b=\s,s_i;\sigma_{-i}^*)\ \ge\  EU_i^\chi((\b_{-i}=\s_{-i},b_i),s_i;\sigma_{-i}^*),\quad \forall b_i.$$

	A mechanism is cursed ex-post individually rational (\cepir) if for every $i, \s,$
	$$EU^\chi_i(\b=\s,s_i)\ge 0.$$ 
	A mechanism that is both \cepic\ and \cepir\ is denoted by \cepicir.
	\end{definition}
	Obviously, \cepic\ implies the interim IC for agents with parameter $\chi$.

	Lemma~\ref{lem_cepicir} introduces an equivalent definition of \cepicir, which simplifies the analysis of whether a mechanism satisfies \cepicir\ or not.
	
	\begin{lemma}
	\label{lem_cepicir}
	%[Cursed ex-post incentive compatibility and individually rationality (\cepicir)]
	A mechanism is \cepic\ if and only if for every $i$, $s_i$ and $\s_{-i}$, $$x_i(\s)v^\chi_i(\s)-p_i(\s)\ge x_i(b_i, \s_{-i})v^\chi_i(\s)-p_i(b_i, \s_{-i}) \quad \forall b_i.$$ 
	A mechanism is \cepir\ if and only if for every $i, \s.$
	$$x_i(\s)v^\chi_i(\s)-p_i(\s)\ge 0.$$ 
	\end{lemma}
	\begin{proof}
	To see this lemma holds, we only need to plug the following expression of the expected utility of bidders into Definition~\ref{def_cepicir}: $\forall \s, b_i$
	\begin{equation}
	\begin{aligned}
	& EU^\chi_i((\b_{-i}=\s_{-i}, b_i), s_i;\sigma_{-i}^*)\\
	= & 
	\int_{\tilde{\s}_{-i}\in S^{n-1}}
	\bigg(
	(1-\chi){f_\sigma}(\tilde{\s}_{-i}|\s_{-i}, s_i)u_i((b_i,\s_{-i}),(s_i,\tilde{s}_i))\\
	&\quad\quad\quad\quad\quad +
	\chi f(\tilde{\s}_{-i}|s_i)u_i((b_i,\s_{-i}),(s_i,\tilde{\s}_{-i}))
	\bigg)
	d\tilde{\s}_{-i}\\
	= & 
	(1-\chi)u_i((b_i,\s_{-i}),\s)\\
	&+
	\chi \int_{\tilde{\s}_{-i}}f(\tilde{\s}_{-i}|s_i)\big(v_i(\tilde{\s}_{-i}, s_i)x_i(b_i,\s_{-i})-p_i(b_i,\s_{-i})\big)d\tilde{\s}_{-i}\\
	= & x_i(b_i,\s_{-i})\left((1-\chi)v_i(\s)+\chi\int_{\tilde{\s}_{-i}}f(\tilde{\s}_{-i}|s_i)v_i(s_i,\tilde{\s}_{-i})d\tilde{\s}_{-i}\right) -p_i(\b_i,\s_{-i})\\
	= & x_i(b_i,\s_{-i})\left((1-\chi)v_i(\s)+\chi\E_{\tilde{\s}_{-i}\sim F|_{s_i}}[v_i(\tilde{\s}_{-i},s_i)]\right) -p_i(b_i,\s_{-i})\\
	= & x_i(b_i, \s_{-i})v_i^\chi(\s) -p_i(b_i,\s_{-i}),
	\end{aligned}
	\end{equation}
    where $v_i^\chi(\s)$ in the last equation is the cursed valuation function of the item, as defined in Eq~(\ref{eq_cursed_valution}).
	\end{proof}

	Setting $\chi=0$ in the definition of \cepic\ gives us the definition of ex-post IC (\epic), where bidders truthfully reporting their signals forms an ex-post Nash equilibrium w.r.t. their true ex-post utilities. It is the strongest incentive guarantee one can hope for in the IDV setting. Similarly, \cepic\ is also the strongest incentive notion we can hope for with cursed agents in the IDV setting. Furthermore, Proposition~\ref{prp_robust} shows that \cepic\ is robust to small estimation errors of the $\chi$ parameter.
	
	\begin{proposition}
    \label{prp_robust}
    Let mechanism $M$ be \cepic\ under cursedness parameter $\chi$, and let agent $i$'s be a $\chi_i$-cursed agent, where  $\chi_i=\chi+\epsilon_i$. The truthful-reporting strategy $\sigma^*$ forms an approximate ex-post Nash equilibrium for agent $i$ with parameter $\chi_i$ in the sense that 
    $$EU_i^{\chi_i}(\b=\s,s_i;\sigma_{-i}^*)\ge EU_i^{\chi_i}((\b_{-i}=\s_{-i}, b_i),s_i;\sigma_{-i}^*)-\epsilon_i\cdot v_i(\bar{s},...,\bar{s})\quad\forall i, \s, b_i.$$
    \end{proposition}

 	 %\paragraph{Ex-post IC (\epic) and Ex-post IR (\epir)} Setting $\chi=0$ in the definition of \cepic\ and \cepir\ gives us the \emph{standard definition} of \epic\ and \epir, respectively.  \epic\ ensures that each agent truthfully reporting their true signal forms an ex-post Nash equilibrium, i.e., no agent can improve their utility by changing their bid given all other bidders bid truthfully. It is the strongest incentive guarantee one can hope for in the IDV setting. By analogy, \cepic\ is also the strongest incentive notion we can hope for with cursed agents in the IDV setting. \epir\ 
 	  %ensures that no bidder will get a negative utility at the truthful equilibrium. 

 	  \paragraph{Ex-post IR (\epir)} 
 	  Setting $\chi=0$ in the definition of \cepir\  gives us the standard definition of \epir, which ensures that no bidder will get a true negative ex-post utility at the truthful-reporting equilibrium.
 	  A mechanism that is \cepicir\ has the outcome that every agent bidding their true signal is a cursed equilibrium (or an ex-post equilibrium in terms of their cursed utilities) with each agent obtaining a non-negative  utility based on their cursed valuation functions.
 	  However, although the agents think their utility will be non-negative for any possible realization of signals $\s$ according to their belief, they might end up paying more than their value for the item leading to a negative utility, because their belief is inaccurate. Therefore, in addition to requiring $\cepicir$, we further consider designing mechanisms that are \epir. Such mechanisms guarantee the agent will not experience \textit{actual} negative utility upon receiving an item, therefore, such an agent would not regret participating in the auction in hindsight.%  \red{such that the winner will not regret joining in the auction and paying the price, even if all the information is revealed after the auction is realized.}

	 \paragraph{Ex-post budget balance (\epbb)} In order to achieve the $\epir$ property, the mechanism might need to make positive transfers since the agents over-estimate their value for the item sold. In order to ensure the seller does not end up with negative revenue, we may also want to require that the mechanisms will satisfy the ex-post budget balance constraint.
	
	\begin{definition}[Ex-post budget-balance]
		A mechanism $M=(x,p)$ is \textit{ex-post budget-balanced} (\epbb) if for every  signal profile $\s$, $\sum_i p_i(\s)\ge 0.$
	\end{definition}
	
	A more relaxed 	requirement is Ex-ante budget-balance, where the mechanism does not lose money \textit{in expectation}.
	%\begin{definition}[Ex-ante Budget-balance]
	%	A mechanism $M=(x,p)$ is budget-balanced if $$\int_{\s}f(\s)\sum_i p_i(\s)d\s \ \geq\  0.$$
	%\end{definition}
	
%	A natural way to ensure budget-balance is to require that the mechanism is \textit{ex-post} budget-balanced.

	When devising a mechanism that satisfies \cepicir, there is a natural tension between \epir\ and budget-balance. The socially optimal mechanism might have negative revenue when satisfying $\epir$ (see Section~\ref{sec:no-budget-balance})). Moreover, while typical mechanisms usually have more revenue with cursed agents (without imposing \epir)~\citep{eyster2005cursed}, when requiring the mechanism to satisfy \epir, the revenue only decreases (see Proposition~\ref{prop:rev-mon}).
	%might shrinks significantly (see Section TBD).
	
	\section{Preliminaries}
	\label{sec:prelims}
	\subsection{\cepicir\ Mechanisms and Virtual Valuations}\label{sec:virtual-value}
	
	\citet{RoughgardenT16} extend Myerson's Lemma and payment identity for the IDV model.  Whenever $v_i^\chi$ is monotone, a simple adaptation of their results characterizes the space of \cepicir\ mechanisms. The proof is omitted, as it is identical to the one in~\cite{RoughgardenT16} for the case of non-cursed agents. 
	\begin{proposition}
		\label{thm:CEPICIR}
		A mechanism $M=(x,p)$ is \cepicir\ if and only if for every $i$, $s_{-i}$, the allocation rule $x_i$ is monotone non-decreasing in the signal $s_i$, and the following payment identity and payment inequality hold:
		\begin{align}p_i(\s) = x_i(\s)v_i^\chi(\s)-\int_{v_i^\chi(0,\s_{-i})}^{v_i^\chi(\s)}x_i((v_i^\chi)^{-1}(t|\s_{-i}), \s_{-i})dt - (x_i(0,\s_{-i})v_i^\chi(0,\s_{-i}) - p_i(0, \s_{-i}))\label{eq:payment identity}
    	\end{align}
		\begin{align}
			p_i(0, \s_{-i})\le  x_i(0,\s_{-i})v^\chi_i(0,\s_{-i})%\jt{(\text{SHOULD BE JUST }	p_i(0, \s_{-i})\le0)}
			\label{eq:cepir}
		\end{align}
	\end{proposition}
	
	 We show that indeed $v_i^\chi$ is monotone in our setting (affiliated signals and monotone $v_i$). We defer the proof to Appendix~\ref{sec:missing-proofs}.
		\begin{lemma}
			\label{lemma_monotone}
			The cursed valuation for agent $i$, $v_i^\chi(\s)$, is monotone-non-decreasing in all agents' signals and monotone-increasing in $s_i$.
		\end{lemma}

		%To prove Lemma~\ref{lemma_monotone}, 
		%we only need to show that $\E_{\s_{-i}|s_i}[v_i(\s)]$ is non-decreasing in $s_i$, because  $v_i^\chi(\s) = (1-\chi)v_i(\s)+\chi\E_{\s_{-i}|s_i}[v_i(\s)]$ and $v_i(\s)$ satisfies the monotonicity condition.  $\E_{\s_{-i}|s_i}[v_i(\s)]$ is non-decreasing in $s_i$ due to 
		%that $v_i(\s)$ is non-decreasing in all signals and signals are also affiliated (positive correlated). We give a formal proof in the appendix.} %Given Lemma~\ref{lemma_monotone}, 
		% We directly give the characterization result as follows.}

	In the setting where valuations are not cursed, setting $p_i(0, \s_{-i})=0$ maximizes the seller's revenue, and makes sure that the seller never has to pay the buyers participating in the auction, therefore ensures that the mechanism is budget-balanced. However, for cursed agents, even though $v_i^\chi(\s)\geq p_i(\s)$, it might as well be the case that $v_i(\s) < p_i(\s),$ resulting in negative utility, and breaching the $\epir$ property. Therefore, fixing a mechanism, one might want to set $p_i(0, \s_{-i})$ to be strictly smaller than zero for some values of $\s_{-i}$, which means the mechanism might pay agents for participating. 
	Thus, in designing a mechanism to guarantee $\epir$, one must take care in order not to violate budget balance.

		\citet{RoughgardenT16} extend the definition of a virtual valuation to interdependent values setting. Given $\s_{-i}$, they define a function
	$$\varphi_i(s_i | \s_{-i})  = v_i(\s) - {v_i}'(s_i, \s_{-i})\frac{1-F(s_i\ |\ \s_{-i})}{f(s_i\ |\ \s_{-i})},$$
	and show that similarly to the private value setting, revenue maximization reduces to virtual welfare maximization. The definition of virtual valuations and formulating revenue maximization as virtual welfare maximization naturally to the case of cursed bidders.
	\begin{definition}[Cursed virtual value] The cursed virtual valuation of agent $i$ conditioned on $\s_{-i}$ is defined as 
		$$\varphi_i^\chi(s_i\ |\ \s_{-i})  = v_i^\chi(\s) - {v_i^\chi}'(s_i, \s_{-i})\frac{1-F(s_i\ |\ \s_{-i})}{f(s_i\ |\ \s_{-i})}.$$
	\end{definition}
	
	The next proposition follows the exact same derivation as the one in~\cite{RoughgardenT16,myerson1981optimal} for non-cursed agents.
	\begin{proposition}[Follows from~\cite{RoughgardenT16,myerson1981optimal}]
		For every interdependent values setting, the expected revenue of a \cepicir\ mechanism equals its expected conditional cursed virtual surplus, up to an additive factor:
		$$\E_{\s}\left[\sum_i p_i(\s)\right]=\E_{\s}\left[\sum_i x_i(\s)\varphi^\chi(s_i|\s_{-i})\right]-\sum_{i}\E_{\s_{-i}}[x_i(0,\s_{-i})v_i^\chi(0,\s_{-i}) - p_i(0, \s_{-i})]$$\label{thm:virtual-welfare}
	\end{proposition}
	
	\subsection{Deterministic \cepicir\ Mechanisms}
	\label{sec_det_cepicir}
	
	In this paper we focus on deterministic mechanisms, as deterministic mechanisms are optimal for our setting whenever bidders are not cursed~\citep{RoughgardenT16}.
	The following is a direct corollary of the monotonicity of \cepicir\ mechanisms. 
	\begin{corollary}
		Any deterministic \cepicir\  mechanism is a threshold mechanism. i.e., for every~$i$, there exists a function $t_i(\cdot)$ such that $x_i(s_i,\s_{-i})=\begin{cases}& 1 \quad s_i> t_i(\s_{-i})\\ & 0 \quad s_i \le t_i(\s_{-i})\end{cases}.$ \label{cor:det-threshold}  
	\end{corollary}
	We refer to $t_i(\s_{-i})$ as the \textit{critical bid} for agent $i$.  Note that when $t_i(\s_{-i})=\bar{s}$, we never allocate to agent $i$.
	The following lemma restricts the set of allocation rules we inspect.
		%The following lemma shows that we can use Lemma~\ref{lem:epir-implies-cepir} in the setting studied in this section.
	\begin{lemma}
		For every deterministic, anonymous $\cepic$ mechanism and for every $\s$, if the items is allocated, it is allocated to a bidder in $\mathrm{argmax}_i\{s_i\}$. \label{lem:max-alloc}
	\end{lemma}
	\begin{proof}
		Assume the item is given to an agent $j$ such that there exists $i$ for which $s_i$ is strictly bigger than $s_j$. By anonymity, there exists some $i\in \mathrm{argmax}_\ell\{s_\ell\}$ such that if we switch $i$ and $j$'s signal, $i$ wins the item. Since $j$ wins at $\s$, $j$ also wins at $\s'=(s'_j=s_i, \s'_{-j} = \s_{-j})$ by monotonicity of \cepic\ mechanisms. Since $i$ wins at $\s^*=(s_i^*=s_j, s_j^*=s_i, \s_{-ij})$, by monotonicity, $i$ also wins at $\s^*=(s_i^*=s_i, s_j^*=s_i, \s_{-ij})=\s'$, a contradiction.   
	\end{proof}
	
	The above lemma implies that when dealing with such mechanisms, assuming Assumptions~\ref{assump:zero-no-alloc} and~\ref{assump:tie-no-alloc} are without loss. By the above lemma, a zero signal cannot win unless it is not the signal, therefore Assumption~\ref{assump:zero-no-alloc} follows from Assumption~\ref{assump:tie-no-alloc}. For assumption~\ref{assump:tie-no-alloc}, one can take any mechanism that violates this assumption, and set $x_i(\s)=0$ for every $\s$ where the highest bid is not unique. By the above lemma, such mechanism remains monotone non-decreasing in a bidder's own bid, therefore $\cepic$.  By continuity of the signal distribution, the resulting mechanism has the same expected revenue and social welfare as the original one.

\section{Implications of Ex-post IR}\label{sec:epir-implication}

In this section we discuss implications of imposing \epir\ on the mechanism. %Missing proofs of this section appear in Appendix~\ref{sec:epir-implication-proofs}
We first show that in order to achieve our incentive properties, it suffices to design an ex-post IR mechanism, and cursed ex-post IR will follow.

	\begin{lemma}
		For every interdependent value setting, %every mechanism in which for each $\s_{-i}$, $x_i(0,\s_{-i})=0$,
		\cepic\ and \epir\ implies \cepir.\label{lem:epir-implies-cepir}
	\end{lemma}
	\begin{proof}
		For a mechanism to be \epir, we need that the actual value an agent gets from an allocation is higher then the price she pays. That is,  $x_i(\s)v_i(\s) \geq p_i(\s)$ for every $\s$. Using equation~\eqref{eq:payment identity} for \cepic\ mechanisms, and rearranging, we get that for every $i$ and every $\s$
		\begin{eqnarray}
			p_i(0, \s_{-i}) \le x_i(\s)\left(v_i(\s) - v_i^\chi(\s)\right)
			+ \int_{v_i^\chi(0,\s_{-i})}^{v_i^\chi(\s)}x_i((v_i^\chi)^{-1}(t|\s_{-i}), \s_{-i})dt  + x_i(0,\s_{-i})v_i^\chi(0,\s_{-i}). \label{eq:epir}
		\end{eqnarray}
		
		Specifically, fixing $\s_{-i}$ and setting $s_i=0$, we get 
		\begin{eqnarray*}
			p_i(0, \s_{-i}) &\le& x_i(0, \s_{-i})\left(v_i(0,\s_{-i}) - v_i^\chi(0,\s_{-i})\right)
			+ \int_{v_i^\chi(0,\s_{-i})}^{v_i^\chi(0,\s_{-i})}x_i((v_i^\chi)^{-1}(t|\s_{-i}), \s_{-i})dt  + x_i(0,\s_{-i})v_i^\chi(0,\s_{-i})\\
			& = & x_i(0,\s_{-i})v_i^\chi(0,\s_{-i}),
		\end{eqnarray*}
		where the inequality used Assumption~\ref{assump:zero-no-alloc}, %\AE{Currently, assumption 1 is removed. Should we make it more visible?}, 
		that agents aren't allocated at their lowest signal ($x_i(0,\s_{-i})=0$). This coincides with Equation~\eqref{eq:cepir}, implying \cepir.
	\end{proof}

	According to Corollary~\ref{cor:det-threshold}, every deterministic allocation rule is equivalent to a set of threshold functions $\{t_i(\cdot)\}_i$. As noted before, the only freedom one have in setting payments of \cepic\ mechanisms is by setting the term $p_i(0,\s_{-i})$. We show that when maximizing revenue subject to \cepic\ and \epir\ constraints, there is a single optimal way to set $p_i(0,\s_{-i})$. % Recall the definition of critical and critical cursed valuations (Definition~\ref{def:critical-values}), we have the following.

	\begin{lemma} 
		Fixing threshold functions $\{t_i(\cdot)\}_i$ of a deterministic anonymous mechanism, the revenue optimal \cepicir\ and \epir\ mechanism sets  $$p_i(0,\s_{-i}) =
		\begin{cases}\min\left\{0, v_i(t_{i}(\s_{-i}),\s_{-i})-v_i^\chi(t_i(\s_{-i}),\s_{-i})\right\} \qquad & \text{if } t_i(\s_{-i})<\bar{s},\\
		0 \qquad & \text{otherwise}.
		
		\end{cases}
		$$  (and therefore, the payment is uniquely defined using Equation~\eqref{eq:payment identity}.)\label{lem:compensation}
	\end{lemma}	
	\begin{proof}
		Since $x_i(0,\s_{-i})=0$ by Assumption~\ref{assump:zero-no-alloc}, we have to have $p_i(0,\s_{-i})\leq 0$, otherwise an agent pays without getting allocated, which leads to negative utility. Consider $\s_{-i}$ such that $t_i(\s_{-i})<\bar{s}$. By Equation~\eqref{eq:epir} for \epir\ (which also implies \cepir), we get that in order to maximize revenue, one should set
		\begin{eqnarray*}
			p_i(0,\s_{-i}) &= & \min\left\{0, \min_{s_i}\left\{x_i(\s)\left(v_i(\s) - v_i^\chi(\s)\right)
			+\int_{v_i^\chi(0,\s_{-i})}^{v_i^\chi(\s)}x_i((v_i^\chi)^{-1}(t|\s_{-i}), \s_{-i})dt \right\} \right\}\\
			&= & \min\left\{0, \min_{s_i>t_i(\s_{-i})}\left\{v_i(\s)-v_i^\chi(\s) + \int_{v_i^\chi(t_i(\s_{-i}),\s_{-i})}^{v_i^\chi(\s)}1 dt \right\} \right\}\\
			& = & \min\left\{0, \min_{s_i>t_i(\s_{-i})}\left\{v_i(\s)-v_i^\chi(\s) + (v_i^\chi(\s)-v_i^\chi(t_i(\s_{-i}),\s_{-i}))\right\}\right\}\\
			&= & \min\left\{0, \min_{s_i>t_i(\s_{-i})}\left\{v_i(\s)-v_i^\chi(t_i(\s_{-i}),\s_{-i})\right\}\right\}\\
			&= & \min\left\{0, v_i(t_{i}(\s_{-i}),\s_{-i})-v_i^\chi(t_i(\s_{-i}),\s_{-i})\right\}\label{eq:compensation}%\\
			%	& = & \min\{0, c_i(\s_{-i})-c_i^\chi(\s_{-i})\},\label{lem:compensation}
		\end{eqnarray*}
		where the second equality follows from the fact that $x_i(\s)=1$ only if $s_i> t_i(\s_{-i})$.
		
		For $t_i(\s_{-i})=\bar{s}$,  agent $i$ will never get the item and therefore, will not incur the winner's curse. Hence, we can set $p_i(0,\s_{-i})=0$.
		%, and the last inequality follows Definition~\ref{def:critical-values}.
	\end{proof}

	We get that when agents experience the winner's curse at the ``critical value'' (meaning their real value is smaller than their perceived value), they pay their real value at the critical bid, while if they do not experience the winner's curse, they pay their cursed value.
	We get the following corollary.
    \begin{corollary}
	\label{lem_price_shorthand}
	Fixing an anonymous, deterministic, \cepicir, \epir\ mechanism with threshold function $t(\cdot)$, the revenue optimal way to set $p_i(0,\s_{-i})$ gives the following payment function:
	$$p_i(\s)=\begin{cases}  p_i(0,\s_{-i})\qquad & s_i\le t_i(\s_{-i})\\  v_i(t_i(\s_{-i}), \s_{-i})\qquad & s_i>t_i(\s_{-i})\text{ and } v_i(t_{i}(\s_{-i}),\s_{-i})-v_i^\chi(t_i(\s_{-i}),\s_{-i})\le0\\  v_i^\chi(t_i(\s_{-i}), \s_{-i})\qquad & s_i>t_i(\s_{-i})\text{ and } v_i(t_{i}(\s_{-i}),\s_{-i})-v_i^\chi(t_i(\s_{-i}),\s_{-i})> 0
		\end{cases}.$$
	\end{corollary}
	
	\begin{proof}
		If $s_i \leq t_i(\s_{-i})$, then $i$ does not get the item, and by Equation~\eqref{eq:payment identity}, $p_i(\s)=p_i(0,\s_{-i})$. If $s_i \geq t_i(\s_{-i})$, then $i$ gets the item, and according to Equation~\eqref{eq:payment identity} and Lemma~\ref{lem:compensation}, 
		\begin{eqnarray*}
    		p_i(\s) & = & v_i^\chi(t_i(\s_{-i}), \s_{-i}) + p_i(0,\s_{-i})\\
    		& = & v_i^\chi(t_i(\s_{-i}), \s_{-i}) + \min\{0, v_i(t_i(\s_{-i}), \s_{-i}) - v_i^\chi(t_i(\s_{-i}), \s_{-i})\}.
		\end{eqnarray*}
		Therefore, if $v_i(t_{i}(\s_{-i}),\s_{-i})-v_i^\chi(t_i(\s_{-i}),\s_{-i})\le 0$, $p_i(\s)= v_i^\chi(t_i(\s_{-i}), \s_{-i}) + 0= v_i^\chi(t_i(\s_{-i}), \s_{-i})$, and if $v_i(t_{i}(\s_{-i}),\s_{-i})-v_i^\chi(t_i(\s_{-i}),\s_{-i})\ge 0$, then $p_i(\s)= v_i^\chi(t_i(\s_{-i}) + v_i(t_i(\s_{-i})-v_i^\chi(t_i(\s_{-i})=v_i(t_i(\s_{-i}).$
	\end{proof}

	An interesting implication of this corollary is that as opposed to the case where we do not require \epir, the welfare and revenue in the case of non-cursed agents (i.e., $\chi=0$) is at least that of the case where agents are cursed ($\chi>0$). We show this more generally by showing that the welfare and revenue are monotonically non-increasing in $\chi$. In order to show this, we first show that for a fixed mechanism, $\epir$ implies that the revenue decreases as $\chi$ increases.
	% We first prove a lemma that will be useful in showing the claim.
	\begin{lemma} \label{lem:payment-mon}
	    Fix a threshold rule $t$. Then for every $\s$, every $i$, and every $0\leq \chi \leq \chi'\leq 1$, $p_i^{\chi}(\s,t)\ge p_i^{\chi'}(\s,t)$, where $p_i^{\chi}(\s,t)$ is the optimal payment an agent $i$ with cursedness parameter $\chi$ has with signals $\s$.
	\end{lemma}

	\begin{proof}
        We prove by the three cases of Corollary~\ref{lem_price_shorthand}.   
	    
	    \textbf{Case 1:} $s_i\leq t(\s_{-i})$. Then, $p_i^{\chi}(\s,t) = \min\{0, v_i(t(\s_{-i}), \s_{-i}) - v_i^\chi(t(\s_{-i}), \s_{-i})\}$, and $p_i^{\chi'}(\s,t) = \min\{0, v_i(t(\s_{-i}), \s_{-i}) - v_i^{\chi'}(t(\s_{-i}), \s_{-i})\}$. Notice that if  $p_i^{\chi}(\s,t) < 0$, then
	    \begin{eqnarray}
	        v_i(t(\s_{-i}), \s_{-i}) & < & v_i^\chi(t(\s_{-i}), \s_{-i})\nonumber\\
	        & = &  (1-\chi)v_i(t(\s_{-i}), \s_{-i}) + \chi\E_{\t_{-i}\sim F_{|t(\s_{-i})}}[v_i(t(\s_{-i}),\t_{-i})]\nonumber\\
	        \iff  v_i(t(\s_{-i}), \s_{-i}) &< & E_{\t_{-i}\sim F_{|t(\s_{-i})}}[v_i(t(\s_{-i}),\t_{-i})]\nonumber\\
	        \iff v_i^\chi(t(\s_{-i}), \s_{-i})
	        & = &  (1-\chi)v_i(t(\s_{-i}), \s_{-i}) + \chi\E_{\t_{-i}\sim F_{|t(\s_{-i})}}[v_i(t(\s_{-i}),\t_{-i})]\nonumber\\
	        &\leq & (1-\chi')v_i(t(\s_{-i}), \s_{-i}) + \chi'\E_{\t_{-i}\sim F_{|t(\s_{-i})}}[v_i(t(\s_{-i}),\t_{-i})]\nonumber\\ & = &
	        v_i^{\chi'}(t(\s_{-i}), \s_{-i}) \nonumber\\ 
	        \iff p_i^{\chi}(\s,t) & = & v_i(t(\s_{-i}), \s_{-i}) - v_i^\chi(t(\s_{-i}), \s_{-i}) \nonumber\\& \ge & v_i(t(\s_{-i}), \s_{-i}) - v_i^{\chi'}(t(\s_{-i}), \s_{-i})\nonumber\\
	        & = & p_i^{\chi'}(\s,t), \label{eq:cursed_deriv}
	    \end{eqnarray}
	    which means that either $p_i^{\chi}(\s,t)=p_i^{\chi'}(\s,t) = 0$, or $p_i^{\chi}(\s,t) \ge p_i^{\chi'}(\s,t)$.
	    
	    \textbf{Case 2:} $s_i > t_i(\s_{-i})$ and $v_i(t_{i}(\s_{-i}),\s_{-i})-v_i^\chi(t_i(\s_{-i}),\s_{-i})\le0$. In this case, according to Corollary~\ref{lem_price_shorthand}, we have $p_i^\chi(\s,t)=v_i(t(\s_{-i}),\s_{-i})$. According to Equation~\eqref{eq:cursed_deriv}, $v_i(t_{i}(\s_{-i}),\s_{-i})-v_i^\chi(t_i(\s_{-i}),\s_{-i})\le0$ implies  $v_i(t_{i}(\s_{-i}),\s_{-i})-v_i^{\chi'}(t_i(\s_{-i}),\s_{-i})\le v_i(t_{i}(\s_{-i}),\s_{-i})-v_i^{\chi}(t_i(\s_{-i}),\s_{-i}) \le 0$, and therefore, $p_i^{\chi'}(\s,t)=v_i(t(\s_{-i}),\s_{-i})$ as well.
	    
	    \textbf{Case 3:} $s_i > t_i(\s_{-i})$ and $v_i(t_{i}(\s_{-i}),\s_{-i})-v_i^\chi(t_i(\s_{-i}),\s_{-i})\ge0$. According to Equation~\eqref{eq:cursed_deriv}, we also have $v_i(t_{i}(\s_{-i}),\s_{-i})-v_i^{\chi'}(t_i(\s_{-i}),\s_{-i})\ge0$, which implies that $p_i^\chi(\s,t)= v_i^\chi(t_i(\s_{-i}),\s_{-i})$ and $p_i^{\chi'}(\s,t)=v_i^{\chi'}(t_i(\s_{-i}),\s_{-i})$. Since by Equation~\eqref{eq:cursed_deriv}, $v_i(t_{i}(\s_{-i}),\s_{-i})-v_i^\chi(t_i(\s_{-i}),\s_{-i})\ge0$ implies $v_i^\chi(t_i(\s_{-i}),\s_{-i})>v_i^{\chi'}(t_i(\s_{-i}),\s_{-i})$, the lemma follows.
	\end{proof}

	We get the following.
	
	\begin{proposition}[Revenue monotonicity] \label{prop:rev-mon}
	    For any $0\leq \chi\leq \chi'\leq 1$ the revenue optimal anonymous deterministic \cepicir\ \epir\ mechanism for $\chi$-cursed agents has revenue at least as high as the revenue optimal anonymous deterministic \cepicir\ \epir\ mechanism for $\chi'$-cursed agents.
	\end{proposition}
	\begin{proof}
		Let $t^{\chi'}$ be the threshold rule of the revenue optimal deterministic \cepicir\ \epir\ mechanism for $\chi$-cursed agents, and let $\REV^\chi(t)$ be the optimal revenue of $\chi$-cursed agents with threshold rule $t$. Then
		\begin{eqnarray*}
			\REV^\chi(t^{\chi'})=\int_{\s}f(\s)\sum_i p_i^\chi(\s,t^{\chi'}) d\s \ge \int_{\s}f(\s)\sum_i p_i^{\chi'}(\s,t^{\chi'}) = \REV^{\chi'}(t^{\chi'}),
		\end{eqnarray*}
		where the inequality follows Lemma~\ref{lem:payment-mon}. Therefore, for the optimal deterministic mechanism for agents with cursedness parameter $\chi$, the revenue can only be larger.
	\end{proof}

	\begin{proposition}[Welfare monotonicity]\label{prop:welfare-mon}
		For any $0\leq \chi\leq \chi'\leq 1$, (a) the welfare optimal deterministic \cepicir\ \epir\ mechanism for $\chi$-cursed agents that satisfies ex-post (ex-ante) budget-balance has  welfare at least as high as the welfare optimal deterministic \cepicir\ \epir\ mechanism for $\chi'$-cursed agents that satisfies ex-post (ex-ante) budget-balance.
	\end{proposition}
	\begin{proof}
		We show that every threshold rule that is ex-post (ex-ante) budget-balanced for $\chi'$-cursed agents is also ex-post (ex-ante) budget-balanced for $\chi$-cursed agents. Since the space of mechanisms is larger for $\chi$-cursed agents, the proposition follows. Fix a threshold rule $t$ and signal profile $\s$. By Lemma~\ref{lem:payment-mon}, we have that $\sum_i p_i^\chi(\s,t) \geq \sum_i p_i^{\chi'}(\s,t)$ and $\int_{\s}f(\s)\sum_i p_i^\chi(\s,t) d\s \ge \int_{\s}f(\s)\sum_i p_i^{\chi'}(\s,t)$. Therefore, if $\sum_i p_i^{\chi'}(\s,t)$ or $\int_{\s}f(\s)\sum_i p_i^{\chi'}(\s,t)$ are non-negative, so are $\sum_i p_i^{\chi}(\s,t)$ or $\int_{\s}f(\s)\sum_i p_i^{\chi}(\s,t)$.
	\end{proof}

	\section{Revenue Maximization}\label{sec:revenue}
	
	In this section, we devise a mechanism that maximizes revenue among all deterministic, anonymous, $\cepicir$ and $\epir$ mechanisms.
	According to Lemma~\ref{lem:max-alloc}, we only consider rules for which $t_i(\s_{-i})\geq \max_{j\neq i} s_j$. We will use $s_{-i}^* = \max_{j\neq i} s_j$
	to denote the second highest signal. We notice that every such mechanism is feasible by definition.
	
	\begin{observation}
		Every deterministic mechanism such that $x_i(\s)=1$ implies that $s_i> s^*_{-i}$ is feasible, meaning that for every $\s$, $\sum_i x_i(\s)\leq 1$. \label{obs:feas}
	\end{observation}
	\begin{proof}
		This is simply because there cannot be two bidders with a signal strictly larger than all other signals.
	\end{proof}

	Therefore, when designing a revenue optimal deterministic anonymous mechanism, one needs to only care about a threshold rule $t(\cdot)$ satisfying $t(\s_{-i})> s^*_{-i}$ (the same for all bidders due to anonymity), and not worry about feasibility or how to set $p_i(\s_{-i})$, as those are set according to Lemma~\ref{lem:compensation}.
	We now show how to devise the optimal anonymous deterministic mechanism.

	\begin{theorem}
		Let \begin{eqnarray}
			r_i(t,\s_{-i})= \begin{cases}
				\min\{v_i^\chi(t,\s_{-i}), v_i(t,\s_{-i})\} - v_i^\chi(t,\s_{-i})F(t|\s_{-i}) \qquad & t\in[s_{-i}^*,\bar{s})\\
				 0 \qquad & t=\bar{s} 
			\end{cases}
		\end{eqnarray}
	
		The optimal deterministic anonymous mechanism sets a threshold $$t(\s_{-i}) \in \arg\max_{t\in[s_{-i}^*,\bar{s}]} r_i(t,\s_{-i}).$$
		
		%$$r_i(t,\s_{-i})=\min\{v_i^\chi(t,\s_{-i}), v_i(t,\s_{-i})\} - v_i^\chi(t,\s_{-i})F(t|\s_{-i}).$$
		%The optimal deterministic anonymous mechanism sets a threshold $$t(\s_{-i}) \in \arg\max_{t\in[s_{-i}^*,\bar{s}]}~  r_i(t,\s_{-i}) \text{ if } \max_{t\in[s_{-i}^*,\bar{s}]} r_i(t, \s_{-i})>0,$$ otherwise, sets $t(\s_{-i})=\bar{s}$
		\label{thm:rev-optimal}
	\end{theorem}

% 	\begin{theorem}
% 		The optimal deterministic anonymous mechanism sets a threshold 
% 		\begin{eqnarray}
% 			t(\s_{-i})=\mathrm{argmax}_{s_{-i}^*\le t<\bar{s} } \min\{v_i^\chi(t,\s_{-i}), v_i(t,\s_{-i})\} - v_i^\chi(t,\s_{-i})F(t|\s_{-i}).\label{eq:opt_threshold}
% 		\end{eqnarray}
	
% 		\label{thm:rev-optimal}
% 	\end{theorem}
	\begin{proof}
	    Using Assumption~\ref{assump:zero-no-alloc} to simplify the expression in Proposition~\ref{thm:virtual-welfare}, our objective is to maximize the following expected revenue.
	    
	    $$\E_{\s}\left[\sum_i x_i(\s)\varphi_i^\chi(s_i|\s_{-i})\right]+\sum_{i}\E_{\s_{-i}}[p_i(0, \s_{-i})].$$
	    
	    Since we are interested in anonymous mechanisms, this is equivalent to maximizing 
	    $$\E_{\s}\left[ x_i(\s)\varphi_i^\chi(s_i|\s_{-i})\right]+\E_{\s_{-i}}[p_i(0, \s_{-i})]$$
	    for a given agent $i$.
	    
        Applying Lemma~\ref{lem:compensation}, we first consider the case $t_i(\s_{-i})<\bar{s}$ and aim to find a threshold function $t(\cdot)$ such that the following is maximized, and then consider the case that $t_i(\s_{-i})=\bar{s}$.
		\begin{eqnarray}
			\int_\s f(\s) x_i(\s)\varphi_i^\chi(s_i|\s_{-i})d\s+ \int_{\s_{-i}}f(\s_{-i})\min\left\{0, v_i(t(\s_{-i}),\s_{-i})-v_i^\chi(t(\s_{-i}),\s_{-i})\right\}d\s_{-i} \nonumber\\
			=  \int_{\s_{-i}}f(\s_{-i})\left(\int_{t(\s_{-i})}^1 f(s_i|\s_{-i})\varphi_i^\chi(s_i|\s_{-i})d s_i + \min\left\{0, v_i(t(\s_{-i}),\s_{-i})-v_i^\chi(t(\s_{-i}),\s_{-i})\right\}\right)d\s_{-i}.\label{eq:objective}
		\end{eqnarray} 
		
		We use the following expansion.
		\begin{eqnarray}
			\int_{t(\s_{-i})}^{\bar{s}} f(s_i|\s_{-i})\varphi_i^\chi(s_i|\s_{-i})d s_i  
			& = & \int_{t(\s_{-i})}^{\bar{s}} f(s_i|\s_{-i})\left(v_i^\chi(\s)-{v_i^\chi}'(\s)\frac{1-F(s_i|\s_{-i})}{f(s_i|\s_{-i})}\right)d s_i\nonumber\\
			& = & \int_{t(\s_{-i})}^{\bar{s}} f(s_i|\s_{-i})v_i^\chi(\s)+F(s_i|\s_{-i}){v_i^\chi}'(\s) d s_i - \int_{t(\s_{-i})}^{\bar{s}} {v_i^\chi}'(\s)d s_i\nonumber\\
			& = & \left(F(s_i|\s_{-i})-1\right)v_i^\chi(\s)\Big\vert_{t(\s_{-i})}^{\bar{s}}\nonumber\\
			& = & \left(F(\bar{s}|\s_{-i})-1\right)v_i^\chi(\bar{s},\s_{-i}) - \left(F(t(\s_{-i})|\s_{-i})-1\right)v_i^\chi(t(\s_{-i}),\s_{-i})\nonumber\\
			& = & (1-F(t(\s_{-i})|\s_{-i}))v_i^\chi(t(\s_{-i}),\s_{-i}),\label{eq:virtual-expansion}
		\end{eqnarray}
		where the third equality uses integration by parts, and the last inequality is due to $F(\bar{s}|\s_{-i})=1$. Plugging Equation~\eqref{eq:virtual-expansion} into Equation~\eqref{eq:objective}, we get:
		\begin{eqnarray}
			\int_{\s_{-i}}f(\s_{-i})\left((1-F(t(\s_{-i})|\s_{-i}))v_i^\chi(t(\s_{-i}),\s_{-i})+ \min\left\{0, v_i(t(\s_{-i}),\s_{-i})-v_i^\chi(t(\s_{-i}),\s_{-i})\right\}\right)d\s_{-i}\label{eq:objective-expanded}
		\end{eqnarray}
		
		Notice that for each $\s_{-i}$, the choice of threshold $t(\s_{-i})$ is independent of a different $\s'_{-i}$'s threshold (we are guaranteed to be feasible by Observation~\ref{obs:feas}). Let $t=t(\s_{-i})$. We wish to choose $t\geq s_{-i}^*$ that maximizes:
		\begin{eqnarray}
			(1-F(t|\s_{-i}))v_i^\chi(t,\s_{-i})+ \min\left\{0, v_i(t,\s_{-i})-v_i^\chi(t,\s_{-i})\right\}\nonumber
		\end{eqnarray}
		
		For $t$ such that $v_i(t,\s_{-i})\geq v_i^\chi(t,\s_{-i})$, we get
		\begin{eqnarray*}
			(1-F(t|\s_{-i}))v_i^\chi(t,\s_{-i}) + v_i(t,\s_{-i}) - v_i^\chi(t,\s_{-i}) = v_i(t,\s_{-i}) - v_i^\chi(t,\s_{-i})F(t|\s_{-i}),
		\end{eqnarray*}
		while for $t$ satisfying $v_i(t,\s_{-i}) < v_i^\chi(t,\s_{-i})$, we have
		\begin{eqnarray*}
			(1-F(t|\s_{-i}))v_i^\chi(t,\s_{-i}) + 0 = v_i^\chi(t,\s_{-i}) - v_i^\chi(t,\s_{-i})F(t|\s_{-i}).
		\end{eqnarray*}
		
		Hence, for every $\s_{-i}$, we should choose a $t> s_{-i}^*$ that maximizes
		\begin{eqnarray}
		\label{eq_op_threshold}
			\min\{v_i^\chi(t,\s_{-i}), v_i(t,\s_{-i})\} - v_i^\chi(t,\s_{-i})F(t|\s_{-i}).
		\end{eqnarray}
		
		Second, we consider the case  $t_i(\s_{-i})=\bar{s}$. In this case, we never allocate to agent $i$ whichever $s_i$ is and $p_i(\s)=p_i(0,\s_{-i})=0$, leading to zero expected revenue. Therefore, Theorem~\ref{thm:rev-optimal} gives the  optimal threshold.

	\end{proof}
	Note that in the case we only require $\cepicir$ without $\epir$, we can set $p_i(0,\s_{-i})=0$, and the optimal mechanism %, which is deterministic and anonymous~\cite{RoughgardenT16, ??},
	chooses a $t>s_{-i}^*$  that maximizes
	\begin{eqnarray*}
		v_i^\chi(t,\s_{-i}) - v_i^\chi(t,\s_{-i})F(t|\s_{-i}),
	\end{eqnarray*}
	which coincides with the mechanism in~\cite{RoughgardenT16} for cursed valuation, while the optimal $\epicir$ mechanism in~\cite{RoughgardenT16} for non-cursed valuations will choose $t>s_{-i}^*$ that maximizes
	\begin{eqnarray*}
		v_i(t,\s_{-i}) - v_i(t,\s_{-i})F(t|\s_{-i}),
	\end{eqnarray*}
	which coincides with our mechanism when $\chi=0$.
	
	\section{Welfare Maximization}\label{sec:indep}
	
We consider the objective of maximizing welfare in an $\epir$ and $\epbb$ way. Unless stated otherwise, we assume that the bidders' signals are sampled i.i.d. In Section~\ref{sec:no-budget-balance} we demonstrate the tension that arises between devising a truthful mechanism for cursed agents, and the requirement that the agents would not experience a negative utility. We show a simple example where in the fully efficient mechanism, the mechanism pays the agents much more than it earns. In Section~\ref{sec:masked-mechanism} we present an operation that takes any deterministic mechanism, and makes it an $\epbb$ mechanism by not allocating in scenarios where the mechanisms would be required to make positive transfers. In Section~\ref{sec:epbb-implies-npt} we show that under natural assumptions on the valuation functions, every $\epbb$ mechanism does not make positive transfers. This implies that the masking operation on the socially optimal mechanism yields a welfare optimal $\epbb$ mechanism (see Proposition~\ref{prop:gva-optimal}). We then present two interesting scenarios: (i) when agents' valuation is the $\max$ function, any $\epbb$ mechanism obtains zero welfare, and (Section~\ref{sec:max}) (ii) a family of valuations including weighted-sum valuations and $\ell_p$-norm of signals, where as the number of agents grows large, the expected $\epbb$ welfare approaches $1/2$ of the expected optimal allocation (Section~\ref{sec:sum}). Missing proofs can be found in the Appendix~\ref{sec:missing-proofs}.

\subsection{Welfare Optimal Mechanism is not Budget Balanced}\label{sec:no-budget-balance}
	Consider $n$ bidders with valuations\footnote{Note these are weighted-sum valuations with $\beta=1/2$.} $v_i(\s)=s_i+\frac{1}{2}\sum_{j\neq i} s_j$ and signals drawn independently from $U[0,1]$. Suppose that $\chi=1$, that is, the agents are fully cursed. The welfare optimal $\cepicir$ mechanism gives the item to the agent with the highest value/signal, and charges payments according to Equations~\eqref{eq:payment identity} and~\eqref{eq:epir}. We show that such mechanism must incur a negative revenue of order $\Theta(n\sqrt{n})$. We note that the mechanism does not even satisfy the less restrictive requirement of ex-ante budget-balance.
	\begin{proposition}
		There exists a setting where the welfare optimal mechanism has an expected revenue loss of $\Theta(n\sqrt{n})$.
		\label{prop:negative-rev}
	\end{proposition} 
	\begin{proof}
		For a signal profile $\s$, assuming that $1$ is the highest, agent $1$ pays $v_1^\chi(s_{-1}^*,\s_{-1})$, and the seller pays each bidder $i$ $-p_i(0,\s_{-i})$ subject to the $\epir$ constraints. To minimize their payment, Lemma~\ref{lem:compensation} implies that the seller sets $p_i(0,\s_{-i}) = \min\left\{0, v_i(t_{i}(\s_{-i}),\s_{-i})-v_i^\chi(t_i(\s_{-i}),\s_{-i})\right\}.$ While since the signals are bounded by 1, it is clear that the expected payment of the highest bidder is $O(n)$. We will show that 
		\begin{eqnarray*}
			\E_{\s}\Big[\sum_i-p_i(0,\s_{-i})\Big]& =&  \E_{\s}\Big[\sum_i-p_i(0,\s_{-i})\Big]
			= \E_{\s}\Big[\sum_i\max\left\{0, v_i^\chi(t_{i}(\s_{-i}),\s_{-i})-v_i(t_i(\s_{-i}),\s_{-i})\right\}\Big]\\&=&\Omega(n\sqrt{n}).
		\end{eqnarray*}
		
		Notice that for a signal profile $\s$ and a bidder $i$ with $\chi=1$, 
		\begin{eqnarray*}
			v_i^1(t_{i}(\s_{-i}),\s_{-i})-v_i(t_i(\s_{-i}),\s_{-i})  =  s_{-i}^* +\E_{\tilde{\s}}\left[\frac{1}{2}\sum_{j\neq i}\tilde{s}_{j}\right] - \left(s_{-i}^* + \frac{1}{2}\sum_{j\neq i} s_j \right) = \frac{n-1}{4}-\frac{1}{2}\sum_{j\neq i} s_j.
		\end{eqnarray*}

		We get 
		\begin{eqnarray*}
			\E_{\s}\Big[\sum_i-p_i(0,\s_{-i})\Big]
			& =& \sum_i \E_{\s}\Big[\max\left\{0,\frac{n-1}{4}-\frac{1}{2}\sum_{j\neq i} s_j\right\}\Big]\\
			& =& \sum_i \frac{1}{2}\E_{\s}\Big[\max\left\{0,\frac{n-1}{2}-\sum_{j\neq i} s_j\right\}\Big]\\
			& \approx & \sum_i \frac{1}{2}\E_{x\sim N(\frac{n-1}{2},(n-1)/12 }\Big[ \max\left\{0,\frac{n-1}{2}-x\right\}\Big] \\
			& = & \sum_i \frac{1}{2}\E_{x\sim N(0,(n-1)/12)}\Big[ \max\left\{0,x\right\}\Big] \\
			& = & \sum_i \frac{1}{2} \E_{x\sim N(0,(n-1)/12)}\Big[x\ |\ x\ge 0 \Big]\Pr_{x}[x\ge 0]\\
			%& = & \frac{n}{2}\sqrt{\frac{2(n-1)}{12\pi}}\ =\ \Theta(n\sqrt{n}). \\
			& = & \frac{n}{4}\sqrt{\frac{(n-1)}{24\pi}}\ =\ \Theta(n\sqrt{n}). 
		\end{eqnarray*}
		
		Here, the approximation follows the central limit theorem, the third equality follows by symmetry of the normal distribution, and the fourth equality follows by taking the expected value of a half-normal distribution. 
		
		Since the seller collects $O(n)$ from the buyers, but pays them $\Theta(n\sqrt{n}),$ the proof follows.
	\end{proof}

	%\subsection{$\epbb$ \cepicir\ Mechanisms}
	\subsection{Masked Mechanisms}\label{sec:masked-mechanism}
	
	We define an operation that takes a deterministic mechanism, and imposes no positive transfers ($p_i(0,\s_{-i})=0$), therefore, the masking operation outputs a mechanism that is trivially $\epbb$.
	\begin{definition} \label{def:masking}
	Given a deterministic mechanism with threshold function $t_i(\s_{-i})$ for all $i,\s_{-i}$, let $\mathcal{NC}(\s_{-i})=\{t|t\ge s_i^*, \text{and } v_i(t,\s_{-i})\ge \E_{\tilde{\s}_{-i}|s_i=st}[v_i(t,\tilde{\s}_{-i})] \}$. A masking of the mechanism is a threshold mechanism with a new threshold function 

	$$t'_i(\s_{-i})=
	\begin{cases}
	 %t_i(\s_{-i})\qquad & \text{if } v_i(t_i(\s_{-i}),\s_{-i})\ge \E_{\tilde{\s}_{-i}\sim F_{|t_i(\s_{-i})}}[v_i(t_i(\s_{-i}),\tilde{\s}_{-i})] \\
	 \inf\{t|t\ge t_i(\s_{-i}),t\in\mathcal{NC}(\s_{-i})\} & \text{if } \{t|t\ge t_i(\s_{-i}),t\in\mathcal{NC}(\s_{-i})\}\ne \emptyset\\
	 1\qquad & \text{otherwise}
	\end{cases}.
	$$
	\end{definition}

    The following lemma shows that indeed masking a mechanism results in a mechanism that is $\epbb$. We note that this implication does not assume that the signals are independent (as opposed to the rest of this section).
	\begin{lemma}
		For any deterministic mechanism that can be implemented in a $\cepicir$, $\epir$, its masking can be implemented in a $\cepicir$, $\epir$, and $\epbb$.
	\end{lemma}
	\begin{proof}
		The masking of a deterministic mechanism is still a threshold mechanism, therefore, it is can be implemented in a \cepicir\ and \epir\ given payments that satisfy Equations~\eqref{eq:cepir} and \eqref{eq:epir}. To show that the mechanism can be implemented in an $\epbb$ manner, we show that for every $i$, $\s_{-i}$ and $\chi$,
		 \begin{eqnarray}
			v_i(t_i(\s_{-i}),\s_{-i})\ge \E_{\tilde{\s}_{-i}\sim F_{|t_i(\s_{-i})}}[v_i(t_i(\s_{-i}),\tilde{\s}_{-i})]\label{eq:alloc_condition}
		\end{eqnarray}  
	implies $\min\left\{0, v_i(t_{i}(\s_{-i}),\s_{-i})-v_i^\chi(t_i(\s_{-i}),\s_{-i})\right\}=0$; Therefore, according to Lemma~\ref{lem:compensation}, we can set $p_i(0,\s_{-i})=0$.
		
		We have that
		\begin{eqnarray*}
			v_i^\chi(t_i(\s_{-i}),\s_{-i}) & = & (1-\chi)v_i(t_i(\s_{-i}),\s_{-i}) + \chi\E_{\tilde{\s}_{-i}\sim F_{|t_i(\s_{-i})}}[v_i(t_i(\s_{-i}),\tilde{\s}_{-i})]\\
			& \le & v_i(t_i(\s_{-i}),\s_{-i}),
		\end{eqnarray*}
		where the equality follows Equation~\eqref{eq_cursed_valution}, and the first inequality follows from the condition in Equation~\eqref{eq:alloc_condition}. The lemma follows.
	\end{proof}

    %\AE{Add a sentence about from here on we assume independence}	
	\subsection{Ex-post Budget-Balance Implies No Positive Transfers}\label{sec:epbb-implies-npt}
	
	In the following we show that under natural conditions, every mechanism that is deterministic, anonymous, \cepicir\ \epir\ and ex-post budget-balance has no positive transfers. This will imply that the optimal  deterministic, anonymous, \cepicir\ \epir\ and ex-post budget-balance mechanism is a masking of the generalized Vickrey auction~\citep{maskin1992auctions,ausubel1999generalized}.% for welfare and masking of the Myerson auction for interdependent values~\cite{RoughgardenT16 citations within} (ALON: check official name) for revenue. 
	We then show that for the max function, every masked mechanism  that allocates the item with zero probability must have positive transfers, implying that such a mechanism will never sell in order to impose ex-post budget-balance. This gives an unbounded gap between the optimal welfare for non-cursed agents and for cursed agents. Finally, we show some interesting families of valuations, where one can provably show that the masking of the generalized Vickrey auction approximate the optimal mechanism for non-cursed agents. 
	
	We first introduce lemmas that will be useful in proving the main result of this section.
	\begin{lemma}
	\label{lem_EPBB_1}
		For every anonymous, deterministic, \cepicir\, \epir, and $\epbb$ mechanism, for every $i$ and $\s_{-i}$,   $p_i(0, \s_{-i})<0$ implies $t_{i}(\s_{-i})=s_{-i}^*.$
	   % For anonymous, deterministic, C-EPIC-IR, EP-IR, and EP-BB mechanisms, 
	    %for each $\s_{-i}$, $p_i(0, \s_{-i})<0$ implies $t_{i}(\s_{-i})=s_{-i}^*.$
	\end{lemma}
	\begin{proof}
		We prove by contradiction. Suppose there exists some $\s_{-i}$ such that $p_i(0,\s_{-i})<0$ but $t_i(\s_{-i})>s_{-i}^*$. Then, for $s_i=\frac{1}{2} (s_{-i}^*+t_i(\s_{-i}))$, we have that for the signal profile $\s$ there are no winners --- $i$ has the highest signal, but it is still lower than $i$'s threshold, and by Lemma~\ref{lem:max-alloc}, only the highest signal can win.
		Therefore, we have that $p_j(\s)= p_j(0,\s_{-j})$ for every bidder $j$, and  $\sum_j p_j(\s) = \sum_j p_j(0, \s_{-j})\le p_i(0, \s_{-i})<0$, which violates the ex-post budget-balance property.
		%$p_j(\s)=p_j(0,\s_{-j})$ for all $j$, and $REV(\s)=\sum_j p_j(\s) = \sum_j p_j(0, \s_{-j})\le p_i(0, \s_{-i})<0$, which violates the EP-BB property.
	\end{proof}
	
% 	\begin{lemma}
% 	\label{lem_EPBB_2}
% 	For every anonymous, deterministic, \cepicir\, \epir, and $\epbb$ mechanism, 
% 	    %For anonymous, deterministic, C-EPIC-IR, EP-IR, and EP-BB mechanisms, 
% 	    if there exist $i$ and $\s_{-i}$ such that  $p_i(0, \s_{-i})<0$ \blue{and $\s_{-i}$ has a unique maximum element, denoted as $s_j (j\ne i)$,  then 
% 	    %for $j={\arg\max}_{j'\ne i}\{s_{j'}\}$, %
% 	    for every $s_i\in[0,s_j)$}, we have $t_j(\s_{-ij},s_i)<s_j$. 
% 	\end{lemma}
    \begin{lemma}
	\label{lem_EPBB_2}
	For every anonymous, deterministic, \cepicir\, \epir, and $\epbb$ mechanism, 
	    if there exist $i$ and $\s_{-i}$ such that  $p_i(0, \s_{-i})<0$, then there's a unique bidder $j$ with maximum signal in $\s_{-i}$, and 
	    for every $s_i\in[0,s_j)$, we have $t_j(\s_{-ij},s_i)<s_j$.
	\end{lemma}
	\begin{proof}
	We prove by contradiction. 
	Suppose there are two bidders with the same highest signal in $\s_{-i}$. For signal profile $\s'=(0, \s_{-i})$, according to Assumption~\ref{assump:tie-no-alloc}, no one will be allocated with item and $p_k(\s')=p_k(0, \s'_{-k})\le 0$ for $k$. As $p_i(0,\s'_{-i})=p_i(0,\s_{-i})<0$, therefore, $\sum_k p_k(\s')<0$, violating \epbb. Thus, there is a unique bidder $j$ with the highest signal in $\s_{-i}$.

	Fix $s_i< s_j$, and suppose $t_j(\s_{-ij},s_i)\geq s_j$, then $j$ cannot win the auction since $j$'s signal is no larger than the threshold $j$ faces. By  Lemma~\ref{lem:max-alloc}, no other agents can win the item, as their signals are not the highest one.  Therefore, we have $p_{j'}(\s)=p_{j'}(0,\s_{-{j'}})$ for all ${j'}$, and $\sum_{j'} p_{j'}(\s) = \sum_{j'} p_{j'}(0, \s_{-{j'}})\le p_i(0, \s_{-i})<0$, violating the ex-post budget-balance property.
	\end{proof}

	We now prove the main result of this section, that under a natural conditions, then every $\epbb$ mechanism that satisfies our desired incentive properties makes no positive transfers. The condition fits the basic intuition that as the actual signals of all other bidders but some bidder $i$ get smaller, $i$'s cursedness increases (as typically, $i$ will overestimate others' signal according to original distribution of signals). %The second condition basically assumes $v_i(\mathbf{0})=0$.
	
	\begin{definition}[Cursedness-monotonicity condition]
		A valuation function satisfies the \textit{cursedness-monotonicity} condition if for every $i$, $\s_{-i}$ for which there exists $s_i>\max_{j\ne i}\{s_j\}$ such that $v_i(\s)-v_i^\chi(\s)<0$, then for any ${\s'}_{-i}\preceq\s_{-i}$ \footnote{For two vectors $\s,\t$, $\s\preceq\t$ if $\s$ is coordinate-wise smaller than or equal to $\t$ when we sort the entries in decreasing order.}and any $s'_i \in (\max_{j\ne i}\{s'_j\},\bar{s})$ 
		%which is strictly smaller than $\bar{s}$,
		\footnote{Recall the support of each signal $s_i$ is denoted as $[0,\bar{s}]$.},
		we also have $v_i(\s')-v_i^\chi(\s')<0$.
	\end{definition}
	
	Proposition~\ref{prp_cursedness_monotonicity_examples} below shows that cursedness-monotonicity condition holds for many widely studied valuation functions such as weighted sum valuations~\citep{RoughgardenT16, EdenFFGK19, EdenFTZ21, myerson1981optimal} and max of signals~\citep{bergemann2020countering, bulow2002prices}.
	
	\begin{proposition}
	\label{prp_cursedness_monotonicity_examples}
	The following valuation functions satisfy the cursedness-monotonicity condition:
	\begin{enumerate}
	    \item $v_i(\s)=s_i+\beta\sum_{j\ne i}s_j$. (Weighted-sum valuations.)
	    %\item $v_i(\s)=h(s_i)+g(\s_{-i})$, where $h$ is an increasing function and $g$ is a non-decreasing function in that sense that for any $\s'_{-i}\preceq \s_{-i},~ g(\s'_{-i})\le g(\s_{-i})$.
	    \item $v_i(\s)=\max_{i}\{s_i\}$. (Maximum of signals.)
	\end{enumerate}
	\end{proposition}

	\begin{theorem} \label{thm:no-positive-transfers}
	For every anonymous, deterministic, \cepicir, \epir\ mechanisms, if the valuation function $v_i$ %is continuous, 
	satisfies cursedness-monotonicity, then a mechanism is ex-post-budget balanced if and only if for every $i$, $\s_{-i}$, $p_i(0,\s_{-i})=0$.
	\end{theorem}
	\begin{proof}
		As the ``if'' direction is immediate, we focus on proving the ``only if'' direction. Assume that we use the optimal way to set $p_i(0,\s_{-i})$ as described in Lemma~\ref{lem:compensation}. This is without loss since if the mechanism is not budget-balanced using optimal setting of $p_i(0,\s_{-i})$, it is not budget-balance for every setting of $p_i(0,\s_{-i})$.
		
		We prove by contradiction. Suppose that there exists $\s_{-i}$ such that $p_i(0, \s_{-i})<0$, therefore, by Lemma~\ref{lem:compensation}, $v_i(t_i(\s_{-i}),\s_{-i})-v_i^\chi(t_i(\s_{-i}),\s_{-i})<0$.  Let $z$ be the smallest non-zero signal in the set of all signals but $s_i$. If there is no such signal, we have $p_i(\mathbf{0}) < 0$, and by anonymity, the revenue of the all zero signal profile is $\sum_j p_j(\mathbf{0})= n p_i(\mathbf{0}) < 0$.
		
		Assume, $z$ is not the only non-zero signal in $\s_{-i}$. 
		By Lemma~\ref{lem_EPBB_2}, let $j$ be the agent with the highest signal in $\s_{-i}$ and for any $s_i<s_{-i}^*$, we have $t_j(\s_{-j})<s_j\leq \bar{s}$. Thus, we have for any $s_i<s_{-i}^*$, 
		$p_j(0, \s_{-j})=\min\{0, v_j(t_j(\s_{-j}), \s_{-j})-v_j^\chi(t_j(\s_{-j}),\s_{-j}))\}<0$, where the inequality is due to the cursedness-monotonicity condition, since $\s_{-j}\preceq \s_{-i}$ and $v_i(t_i(\s_{-i}),\s_{-i})-v_i^\chi(t_i(\s_{-i}),\s_{-i})<0$, as stated above. 
		Therefore, we can set $s_i=0$, and have $p_j(0, 0,\s_{-ij}) < 0$. By anonymity, we also have $p_i(0, 0,\s_{-ij})<0$.  
		
		We continue the process iteratively until we reach the signal profile $\s_{-i}=(0,\ldots,0, z)$, we have $p_i(0,\s_{-i})< 0$.  Moreover, by Lemma~\ref{lem_EPBB_2}, $t_i(\mathbf{0}_{-i})<z$. Continuing with the process for another step, we also get that $p_i(\mathbf{0})<0$, which implies $t_i(\mathbf{0}_{-i})=0$ by Lemma~\ref{lem_EPBB_1}.

		Now consider the final signal profile we get, $\s=(s_1=z,0,\ldots,0)$, where agent 1 denotes the agent with the smallest non-zero signal $z$ in the original signal profile. By the above argument, agent 1 wins the auction (as $t_i(\mathbf{0}_{-i}))=0$). Moreover, by Corollary~\ref{lem_price_shorthand}, $p_1(\s)=v_1(t_i(\mathbf{0}_{-i}),\mathbf{0}_{-i})=v_1(\mathbf{0})=0$, while for all other agents $p_i(\s) = p_i(0, \s_{-i}) < 0$, therefore, $\sum_i p_i(\s) < 0$, contradicting ex-post budget-balance.	
	\end{proof}

	Next, we apply the above characterization to determine the welfare-optimal $\epbb$ mechanism. Then, we show that for some valuation function (e.g., the $\max$ function), the welfare-optimal $\epbb$  mechanism attains zero welfare; while for some other valuation functions, including the weighted sum valuations, the welfare-optimal mechanisms are simple and can approximate the full efficiency attained by an $\epbb$ mechanism for fully rational agents (i.e., $\chi=0$).

	\subsection{Optimal Mechanism}
    
    As standard in the interdependent literature, when devising the welfare-optimal mechanism, we assume that the valuations satisfy the single-crossing condition, presented here in the context of symmetric valuation functions.
    
    \begin{definition}[Single-Crossing for symmetric valuation functions]
    Symmetric valuation functions $\{v_i\}_{i\in [n]}$ satisfy the single crossing condition if for any signal profile $\s$ and agents $i,j$, $s_i\ge s_j$ if and only if $v_i(\s)\ge v_j(\s)$. 
    \label{def:single-crossing}
    \end{definition}
    
    %We use this common assumption for the results regarding the characterization and analysis of the welfare-optimal mechanism.~(Theorem~\ref{thm_masked_SW}).

	Theorem~\ref{thm:no-positive-transfers} implies that for valuations that satisfy the single-crossing condition, the following masked version of the generalized Vickrey auction (\GVA) for cursed valuations is the welfare optimal mechanism that satisfies $\epbb$. \GVA\  assigns the item to the bidder with the highest valuation given all reported signals, and charges the winner the valuation of the item at the minimum winning signal for the winner fixing others' reported signals. With the symmetry settings and single-crossing condition, \GVA\ allocates to the bidder with the highest signal. The masking of \GVA\ is defined as follows.
		
		\begin{definition}[Masked Generalized Vickrey Auction (M-GVA)]
            The \textit{masked generalized Vickrey auction} considers the threshold rule  $t'(\s_{-i})$ which is the result of taking the the threshold rule $t(\s_{-i})=\max_{j\neq i} s_j$, and masking it to ensure no positive transfers (Definition~\ref{def:masking}). The payments are set using the payment identity (Equation~\eqref{eq:payment identity}), where $p_i(0,\s_{-i})$ is set according to Lemma~\ref{lem:compensation}.    
		\end{definition}
		
		%welfare optimal mechanism that satisfies ex-post budget-balance is a masked mechanism.
	
	\begin{proposition}
	    M-GVA is the welfare-optimal mechanism among deterministic, anonymous, \cepicir, \epir\ and $\epbb$ mechanisms, for  continuous valuation functions $v_i$ that satisfies cursedness-monotonicity.  \label{prop:gva-optimal}
	    
	\end{proposition}

	\begin{proof}
		By Theorem~\ref{thm:no-positive-transfers}, we have the socially optimal mechanism must be a masked mechanism. Also, notice that the allocation rule of a GVA mechanism can be written as a threshold allocation rule, with threshold function $t_i^\GVA(\s_{-i})=s_i^*.$ Therefore, we only need to prove that the threshold allocation rule $t_i^\Mask(\cdot)$ masking over $t_i^\GVA(\cdot)$ maximizes the social welfare over all valid threshold allocation rules ${t_i}(\cdot)$.
		To see this, if we lower the threshold $t_i(\s_{-i})$ from $t_i^\Mask(\s_{-i})$, then we either violate the feasible constraint that $t_i(\s_{-i})\ge s_{-i}^*$ or violate no positive transfer, i.e., $p_i(0,\s_{-i})=v_i(t_i(\s_{-i}),\s_{-i})-v_i^\chi(t_i(\s_{-i}), \s_{-i})<0$. If we increase the threshold $t_i(\s_{-i})$ from $t_i^\Mask(\s_{-i})$, then the social welfare is decreased as 
		$\E_{\s}[SW(s)]=\sum_i\E_{\s_{-i}}[\E_{s_i|\s_{-i}}[SW(s)]]=\sum_i\E_{\s_{-i}}[\int_{t_i(\s_{-i})}^{\bar{s}}v_i(t,\s_{-i})f(t|\s_{-i})dt]$.

	\end{proof}

	\subsubsection{ $\max$ Function Has Zero Welfare} \label{sec:max}
	For the max function, Theorem~\ref{thm:no-positive-transfers} implies that $\epbb$ mechanisms allocate the item with zero probability, because as long as the winner's signal is not $\bar{s}$, the winner is cursed, i.e., $v_i(\s)<v_i^\chi(\s)$ for winner $i$ and $s_i<\bar{s}$.
	\begin{corollary}
	\label{cor:max}
	    Consider $v_i(\s)=\max_i\{s_i\}$ with signals drawn i.i.d. from $U[0,1]$ and $\chi$-cursed agents for $\chi>0$. Then any deterministic, anonymous, \cepicir, \epir\ and $\epbb$ mechanism allocates with 0 probability (and thus has 0 revenue and welfare).
	\end{corollary}

\subsubsection{Approximate Efficiency} \label{sec:sum}
In contrast to the $\max$ function, many other valuation functions achieve good efficiency guarantees. To demonstrate this, we define a family of valuations functions including well studied functions such as weighted-sums valuations~\citep{RoughgardenT16,klemperer1998auctions,wilson1969communications,myerson1981optimal,EdenFFGK19,EdenFTZ21}, and $\ell_p$ norms for a finite $p$, for which the $\MGVA$ mechanism approximates the fully efficient mechanism.
\begin{definition}[Concave-Sum valuations]
    Concave-Sum valuations are valuations that can be expressed in the form of  $v_i(\s)=l(g(s_i)+\sum_{j\ne i}h(s_j))$, where $g, h, l$ are strictly increasing and bounded functions on the support of $s_i$, and $l$ is concave.
\end{definition}

\begin{theorem}
\label{thm_masked_SW}
For agents with Concave-Sum valuations, if the valuation function $v_i$ satisfies the single-crossing condition, the $\MGVA$ mechanism has welfare that approaches $\frac{1}{2}$ of the optimal social welfare as the number of agents grows large.
\end{theorem}

\begin{proof}

Let $h_i=h(s_i)$ for any $i\in\N$, and $\h=(h_1,...,h_n)$. Let $\bar{h}=\frac{\sum_i h_i}{n}$ and $\bar{h}_{-i}=\frac{\sum_{j\ne i} h_j}{n-1}$. Let $\lambda=\E_{s_i}[h_i]$, and let $b$ be the supremum of $h(\cdot)$ on the support of $s_i$. Note that as $s_i, i\in\N$ are i.i.d., $h_i,i\in\N$ are also i.i.d. Let $i^*$ denote the bidder with highest signals among all agents with tie breaking arbitrarily. Note that because of the single-crossing condition, $v_{i^*}(\s)$ is also no less than any other value $v_j(\s)$ for $j\ne i^*$.

Consider the signal profile $\s$ such that $\bar{h}\ge\lambda+\frac{b}{n}$. We know that 
$$\E_{\s}[SW^\MGVA(\s)]\ge \E\left[SW^\MGVA(\s)|\bar{h}\ge\lambda+\frac{b}{n}\right]\cdot\Pr\left[\bar{h}\ge\lambda+\frac{b}{n}\right].$$ So to prove our theorem, we only need to prove that  $$\lim_{n\rightarrow+\infty}\Pr\left[\bar{h}\ge\lambda+\frac{b}{n}\right]=\frac{1}{2}\text{ and }\E\left[SW^\MGVA(\s)|\bar{h} \ge\lambda+\frac{b}{n}\right]\ge \E\left[SW^\OPT(\s)\right].$$

First, we  prove that $\lim_{n\rightarrow+\infty}\Pr\left[\bar{h}\ge\lambda+\frac{b}{n}\right]=\frac{1}{2}.$ To see this, as $h_i,i\in\N$ are i.i.d., by the central limit theorem, we have when $n\rightarrow +\infty$,  $\sqrt{n}(\bar{h}-\lambda)\rightarrow\mathcal{N}(0,\sigma^2)$ for some fixed $\sigma$. Thus, we have $\lim_{n\rightarrow+\infty}\Pr[\bar{h}\ge\lambda+\frac{b}{n}]=\lim_{n\rightarrow+\infty}1-\Phi(\frac{b}{\sigma\sqrt{n}})=\frac{1}{2}$, where $\Phi(\cdot)$ is the CDF of the standard normal distribution. 

Second, we prove that $\E\left[SW^\MGVA(\s)|\bar{h} \ge\lambda+\frac{b}{n}\right]\ge \E\left[SW^\OPT(\s)\right].$ We also divide the proof  into two parts. In the first part, we show that $\E\left[SW^\MGVA(\s)|\bar{h} \ge\lambda+\frac{b}{n}\right]=\E\left[v_{i^*}(\s)|\bar{h} \ge\lambda+\frac{b}{n}\right]$: note that for any $\s$ such that $\bar{h}\ge\lambda+\frac{b}{n}$, we have $\forall i\in\N$, $\bar{h}_{-i}\ge\lambda$ (because $\forall i, h_i\le b$), and thus, we have that $\forall i\in N$,
%\begin{small}
\begin{align*}
v_i(t^\MGVA_i(\s_{-i}),\s_{-i}) = & l\left(g(t^\MGVA_i(\s_{-i}))+\sum_{j\ne i}h_j\right)\\
\ge & l\left(g(t^\MGVA_i(\s_{-i}))+(n-1)\lambda\right)\\
= & l\left(\E_{\tilde{\s}_{-i}}\left[g(t^\MGVA_i(\s_{-i}))+\sum_{j\ne i}h(\tilde{s}_j)\right]\right)\\
\ge & \E_{\tilde{\s}_{-i}}\left[l\left(g(t_i^\MGVA(\s_{-i}))+\sum_{j\ne i}h(\tilde{s}_j)\right)\right]\\
= & \E_{\tilde{\s}_{-i}}[v_i(t^\MGVA_i(\s_{-i}),\tilde{\s}_{-i})].
\end{align*}
%\end{small}

The last inequality is due to the concavity of $l$ and the Jensen's inequality. 
Then, according to the definition of \MGVA\ mechanism, we have that that if $\s$ satisfies $\bar{h}\ge\lambda+\frac{b}{n}$, the \MGVA\ will allocate the item to agent $i^*$ and produce social welfare $v_{i^*}(\s)$. Thus, we have $\E\left[SW^\MGVA(\s)|\bar{h} \ge\lambda+\frac{b}{n}\right]=\E\left[v_{i^*}(\s)|\bar{h} \ge\lambda+\frac{b}{n}\right]$.

In the second part, we prove that 
$\E\left[v_{i^*}(\s)|\bar{h} \ge\lambda+\frac{b}{n}\right]\ge\E\left[SW^\OPT(\s)\right].$ This proof relies on  Lemma~\ref{lem_recursive_expectation} below, whose proof is given in Appendix~\ref{sec:missing-proofs}.

\begin{lemma}
\label{lem_recursive_expectation}
If $\z$ is a vector of (possibly correlated) random variables, each with the same support $[a, b]$, and $q(\z)$ is a non-decreasing function in $z_i$ for any $i$ given any $\z_{-i}$,
then for any $i$, and a constant $d$,
$$\E_{\z}\left[q(\z)|\sum z_j \ge d\right]\ge \E_{\z_{-i}}\left[\E_{z_i|\z_{-i}}[q(\z)]|\sum_{j\ne i} z_{j}\ge d-b\right].$$
\end{lemma}

Let $\tilde{v}_i(\h)=l(g(h^{-1}(h_i))+\sum_{j\ne i}h_j)$.
Let $h_{(i)}$ be the order statistics of $\h$ in the order that $h_{(1)}\ge ...\ge h_{(n)}$. Let $\h_{(i)^+}=(h_{(i)},...,h_{(n)})$. 
Let $\xi_{(i)}(\h_{(i)^+})=\E_{\h_{(1),...,(i-1)}|\h_{(i)^+}}[l_{(1)}(g(h^{-1}(h_{i})+\sum_{j\ne i} h_j))].$ As $h_i,i\in\N$ are i.i.d., we have for any $i\in[n-1]$ and for any $j>i$, $$F(h_{(i)}|h_{(i+1)},...,h_{(j-1)},h'_{(j)},h_{(j+1)},...,h_{(n)})$$ weakly FOSD\footnote{Random variable $A$ weakly FOSD random variable $B$ if for any outcome $x$, $\Pr[A\ge x] \ge \Pr[B\ge x]$. If $A$ weakly FOSD $B$, then we have for any non-decreasing function $q(\cdot)$, $\E_A[q(A)]\ge \E_{B}[q(B)]$.} $F(h_{(j)}|\h_{(j+1)^+})$ if $h'_{(j)}>h_{(j)}$. 
As $\tilde{v}_i(\h)$ is non-decreasing with $h_i$ for any $i$ given any $\h_{-(i)}$, by mathematical induction and the first of stochastic dominance, we get that for any $i\in\N$,  $\xi_{(i)}(\h_{(i)^+})$ is non-decreasing in $h_{(j)}$ for any $j\ge i$ given any $\h_{-(j)}$. Therefore, we have 

\begin{align*}
\E_\s\left[v_{i^*}(\s)|{\bar{h}}\ge \lambda+\frac{b}{n}\right]
= &\E_{\h}\left[\tilde{v}_{(1)}(\h)\big|\sum_j h(j)\ge n\lambda+b\right]\\
\ge & \E_{\h_{(2)^+}}\left[
    \xi_{(2)^+}(\h_{(2)^+})\big| \sum_{j=2}^n h(j) \ge n\lambda
\right]\\
\ge & \ldots\\
\ge & \E_{\h_{(n)^+}}\left[
    \xi_{(n)^+}(\h_{(n)^+})\big| h_{(n)} 
    \ge n\lambda - (n-1)b
\right]\\
\ge & \E_{\h_{(n)^+}}\left[
    \xi_{(n)^+}(\h_{(n)^+})
\right]\\
= & \E_\h[\tilde{v}_{(1)}(\h)] = \E_\s[v_{i^*}(\s)] = \E_\s[SW^\OPT(\s)].
\end{align*}
The inequalities are because of the monotone-increasing property of $\xi_{(i)}(\cdot)$ and Lemma~\ref{lem_recursive_expectation}.
Combining the two parts, we finally have $\E\left[SW^\MGVA(\s)|\bar{h} \ge\lambda+\frac{b}{n}\right]\ge \E\left[SW^\OPT(\s)\right],$ which concludes our entire proof.

\end{proof}

	\section{Conclusion and Future Directions}
	
	In this work, we studied the design of mechanisms for agents who suffer from overestimating their value, where the seller tries to avoid agents from suffering from the winner's curse. We designed mechanisms that are deterministic and anonymous, while maximizing the revenue and the welfare of the seller without allowing buyers to have a negative utility. For welfare maximization, we added a requirement that the seller would never have a negative revenue \textit{ex-post}. While we devised optimal mechanisms for these settings, there are many dimensions to the problem, where relaxing any one of these dimensions might lead to a new and interesting design problem. For instance, one can ask the question of what happens if we allow for randomized mechanisms? Mechanisms that can discriminate against bidders? Mechanisms that satisfy the budget-balance constraint only ex-ante? Relaxing each of these dimensions, or a combination, will lead to a new and intricate design problem.
	
	Another interesting question one may ask is how much revenue or welfare exactly do we lose by the fact the agents are biased? Can we relate this loss to the cursedness parameter $\chi$? How much do we lose by being `nice' and helping the buyers not lose money, although they are not playing rationally?
	
	We hope this paper opens the way for other studies answering these, and other interesting and related questions.

    %\paragraph{Acknowledgement} We deeply thank Yannai Gonczarowski and Inbal Talgam-Cohen for their helpful comments. The work is partially supported by National Science Foundation under grant IIS-2007951.
	
\section{Acknowledgements}
We thank the support of National Science Foundation under Grant No. IIS-2007887.
	
	% Bibliography
	\bibliographystyle{plainnat}
	\bibliography{CursedSatisfied.bib}
	
	% Appendix
	\appendix

\section{Cursed Equilibrium in the Wallet-Game} \label{sec:curse-example}

In this section, we demonstrate how the winner's curse naturally arises form considering the cursed-equilibrium model.

Recall the wallet game example and suppose that Alice and Bob have $\chi=1$. Then in the $\chi$-cursed equilibrium, Alice bids as if she is a fully rational agent who values the item by $\E_{s_{Bob}\sim U[0,100]}[s_{Alice}+s_{Bob}|s_{Alice}=\$30]=\$80$. Thus, under the second price auction, she bids \$80 and experiences the winner's curse upon winning. \citet{avery1997second} conducted a lab experiment about this wallet game with each agent's wallet money drawn from $U[1,4]$. The best linear regressor of agents' strategy shows that agents bid by $2.64 + 1.13si$, close to the expected valuation $\E_{s_{-i}\sim U[1,4]}[s_i+s_{-i}]=2.5 + s_i$, instead of the BNE strategy $2s_i$, indicating the agents have a $\chi>0$. \citet{eyster2005cursed} further showed that any $\chi>0$ fits data better than the fully rational case ($\chi=0$) and with a 95\% confidence interval of $[0.59, 0.67]$.

\section{Missing proofs}\label{sec:missing-proofs}
\subsection{Proof of Proposition~\ref{prp_robust}}
\begin{proof}
    As the mechanism is \cepic\ under parameter $\chi$, we have 
    $$x_i(\s)v_i(\s)-p_i(\s)\ge x_i(b_i, \s_{-i})v_i(\s)-p_i(b_i,\s_{-i})\quad \forall i, \s, b_i,$$
    Therefore, we have $\forall i, \s, b_i:$
    \begin{align*}
        EU_i^{\chi_i}(\b=\s,s_i;\sigma_{-i}^*) 
        = & 
        x_i(\s)v_i^{\chi'}(\s)-p_i(\s) \\
        = & 
        x_i(\s)((1-\chi-\epsilon)v_i(\s)-(\chi+\epsilon) \E_{\tilde{\s}_{-i}|s_i}[v_i(s_i,\tilde{\s}_{-i})])-p_i(\s)\\
        = & 
        x_i(\s)v_i^\chi(\s) -p_i(\s) +\epsilon\cdot x_i(\s)\left(\E_{\tilde{\s}_{-i}|s_i}[v_i(s_i,\tilde{\s}_{-i})]-v_i(\s)\right)\\
        \ge & 
        x_i(b_i,\s_{-i})v_i^\chi(\s) -p_i(b_i,\s_{-i}) +\epsilon\cdot x_i(\s)\left(\E_{\tilde{\s}_{-i}|s_i}[v_i(s_i,\tilde{\s}_{-i})]-v_i(\s)\right)\\
        = & 
        x_i(b_i,\s_{-i})v_i^{\chi'}(\s) -p_i(b_i,\s_{-i}) \\
        & +\epsilon\cdot x_i(\s)\left(\E_{\tilde{\s}_{-i}|s_i}[v_i(s_i,\tilde{\s}_{-i})]-v_i(\s)\right)
        -\epsilon\cdot x_i(b_i,\s_{-i})\left(\E_{\tilde{\s}_{-i}|s_i}[v_i(s_i,\tilde{\s}_{-i})]-v_i(\s)\right)\\
        = &
        x_i(b_i,\s_{-i})v_i^{\chi'}(\s) -p_i(b_i,\s_{-i}) + \epsilon \big(x_i(\s)-x_i(b_i,\s_{-i})\big)\left(\E_{\tilde{\s}_{-i}|s_i}[v_i(s_i,\tilde{\s}_{-i})]-v_i(\s)\right)\\
        \ge & 
        x_i(b_i,\s_{-i})v_i^{\chi'}(\s) -p_i(b_i,\s_{-i})-\epsilon\cdot v_i(\bar{s},...,\bar{s})\\
        = & 
        EU_i^{\chi_i}((\b_{-i}=\s_{-i}, b_i),s_i;\sigma_{-i}^*)-\epsilon_i\cdot v_i(\bar{s},...,\bar{s})
    \end{align*}
    \end{proof}
    
    \subsection{Proof of Lemma~\ref{lemma_monotone}}
    \begin{proof}
    We only need to prove that for agent $i$, $\E_{\s_{-i}|s_i}[v_i(\s)]$ is non-decreasing in $s_i$. 
    We first prove that for any function $g^{(n)}(\s)$ non-decreasing in each signal, and for affiliated signals $\s=(s_1,...,s_n)$, $g^{(n-1)}(\s_{-j})=\E_{s_j|\s_{-j}}[g^{(n)}(\s)]$ is also non-decreasing in signal $s_i$ for all $i\ne j$. This is because signal affiliation of $\s$ implies that for any pair of $i$ and $j$, fixing $\s_{-ij}$, $s_i$ and $s_j$ are also affiliated, which further implies that $s_j|s_i=x$ weakly first-order stochastic dominates (FOSD) $s_j|s_i=y$  for any $x>y$. As a result of the FOSD and the non-decreasing property of $g^{(n)}(\s)$, we have
    \begin{align*}
    g^{(n-1)}(s_i=x,\s_{-ij})
    &=\E_{s_j|s_{-ij}, s_i=x}\left[g^{(n)}(s_j, \s_{-ij}, s_i=x)\right]\\
    & \ge\E_{s_j|s_{-ij}, s_i=x}\left[g^{(n)}(s_j, \s_{-ij}, s_i=y)\right] \\
    & \ge\E_{s_j|s_{-ij}, s_i=y}\left[g^{(n)}(s_j, \s_{-ij}, s_i=y)\right] 
    = g^{(n-1)}(s_i=y,\s_{-ij})
    \end{align*}
    Therefore, $g^{(n-1)}(\s_{-j})$ is non decreasing in $s_i$ for any $i\ne j$. 
    By induction starting with $g^{(n)}(\s)=v_i(\s)$, we can get $g^{(1)}(s_i)=\E_{\s_{-i}|s_i}[v_i(\s)]$ is non-decreasing in $s_i$. 
    \end{proof}
    
    \subsection{Proof of Proposition~\ref{prp_cursedness_monotonicity_examples}}
    \begin{proof}
    We first prove that 
    $v_i(\s)=s_i+\beta\sum_{j\ne i}s_j$ satisfies the cursedness-monotonicity condition. 
    We have 
    \begin{align*}
    v_i(\s)-v_i^\chi(\s) 
    = & \chi\left(v_i(\s)-\E_{\tilde{\s}_{-i}|s_i}[v_i^\chi(s_i, \tilde{\s}_{-i})]\right)\\
    = & \chi\left(\sum_{j\ne i}s_j -\E_{\tilde{\s}_{-i}|s_i}\left[\sum_{j\ne i}\tilde{s}_j\right]\right)\\
    = & \chi\left(\sum_{j\ne i}s_j -\E_{\tilde{\s}_{-i}}\left[\sum_{j\ne i}\tilde{s}_j\right]\right)\\
    \end{align*}
    The last equation is due to signals are assumed to be independent in this section. 
    Therefore, if for some $\s$, $v_i(\s)<v_i^\chi(\s)$, then we have  $\sum_{j\ne i}s_j <\E_{\tilde{\s}_{-i}|s_i}\left[\sum_{j\ne i}\tilde{s}_j\right]$ and thus for any $\s'_{-i}\preceq \s_i$, $\sum_{j\ne i}s'_j\le\sum_{j\ne i}s_j <\E_{\tilde{\s}_{-i}|s_i}\left[\sum_{j\ne i}\tilde{s}_j\right]$, implying for any $s_i'$, $v_i(\s')-v_i^\chi(\s')<0$, which completes the proof. Also, note that if the signals are not independent but positively affiliated in the sense that $\forall s_i'>s_i, \E_{\tilde{\s}_{-i}|s_i'}[\sum_{j\ne i}\tilde{s}_j]\ge \E_{\tilde{\s}_{-i}|s_i}[\sum_{j\ne i}\tilde{s}_j]$, the cursedness monotinicity still holds. 
    
    % We first prove that 
    % $v_i(\s)=h(s_i)+g(\s_{-i})$ satisfies the cursedness-monotonicity condition. 
    % We have 
    % \begin{align*}
    % v_i(\s)-v_i^\chi(\s) 
    % = & \chi\left(v_i(\s)-\E_{\tilde{\s}_{-i}|s_i}[v_i^\chi(s_i, \tilde{\s}_{-i})]\right)\\
    % = & \chi\left(g(\s_{-i})-\E_{\tilde{\s}_{-i}|s_i}[g(\tilde{\s}_{-i})]\right)\\
    % = & \chi\left(g(\s_{-i})-\E_{\tilde{\s}_{-i}}[g(\tilde{\s}_{-i})]\right)
    % \end{align*}
    % The last equation is due to signals are assumed to be independent in this section. 
    % Therefore, if for some $\s$, $v_i(\s)<v_i^\chi(\s)$, then we have  $g(\s_{-i})<\E_{\tilde{\s}_{-i}}[g(\tilde{\s}_{-i}))]$ and thus for any $\s'_{-i}\preceq \s_i$, $g(\s'_{-i})\le g(\s_{-i})<\E_{\tilde{\s}_{-i}}[g(\tilde{\s}_{-i})]$, implying for any $s_i'$, $v_i(\s')-v_i^\chi(\s')<0$, which completes the proof. Also, note that if the signals are not independent but positively affiliated in the sense that $\forall s_i'>s_i, \E_{\tilde{\s}_{-i}|s_i'}[g(\tilde{\s}_{-i})]\ge \E_{\tilde{\s}_{-i}|s_i}[g(\tilde{\s}_{-i})]$, the cursedness monotinicity still holds. 
    
    Second, we prove that  $v_i(\s)=\max_s\{s_i\}$ satisfies the cursedness-monotonicity condition. This is simply because that $\forall \s_{-i}$ and $s_i\in(\max_{j\ne i}\{s_j\},\bar{s})$, $v_i(\s)-v_i^\chi(\s)<0$. To see this, we have $v_i(\s)=s_i<\E_{\tilde{\s}_{-i}|s_i}[\max\{s_i, \max_{j\ne i}\{\tilde{s_j}\}\}]=\E_{\tilde{\s}_{-i}|s_i}[v_i(s_i,\tilde{\s}_{-i})]$, implying $v_i(\s)-v_i^\chi(\s)<0$. 
    \end{proof}

\subsection{Proof of Corollary~\ref{cor:max}}
	\begin{proof}
	Because $v_i(\s)=\max_i\{s_i\}$, we have
	$\forall \s_{-i}$ and $s_i\in(\max_{j\ne i}\{s_j\},\bar{s})$, $v_i(\s)-v_i^\chi(\s)<0$. To see this, we have $v_i(\s)=s_i<\E_{\tilde{\s}_{-i}|s_i}[\max\{s_i, \max_{j\ne i}\{\tilde{s_j}\}\}]=\E_{\tilde{\s}_{-i}|s_i}[v_i(s_i,\tilde{\s}_{-i})]$, implying $v_i(\s)-v_i^\chi(\s)<0$. 
	Therefore, for any threshold function $t_i(\cdot)$, which satisfies $t_i(\s_{-i})\ge \max_{j\ne i} \{s_j\}, \forall s_{-i}$ according to Lemma~\ref{lem:max-alloc}, we have $v_i(t_i(\s_{-i}),\s_{-i})-v_i^\chi(t_i(\s_{-i}),\s_{-i})\le 0$, where the equality holds only if $t_i(\s_{-i})=\bar{s}$. Therefore, for any $\s_{-i}$, if $t_i(\s_{-i})<\bar{s}$, then we have $p_i(0,\s_{-i})=\min\{0, v_i(t_i(\s_{-i}),\s_{-i})-v_i^\chi(t_i(\s_{-i}),\s_{-i})\}<0$. Since the $\max$ function satisfies the cursedness-monotonicity, Theorem~\ref{thm:no-positive-transfers} implies such $t_i(\cdot)$ cannot be supported by any deterministic, anonymous, \cepicir, and \epbb\ mechanism. Consequently, the only deterministic, anonymous, \cepicir, and \epbb\ mechanism is to either allocate to a bidder with $s_i=\bar{s}$ or never allocate, leading to zero allocation probability and thus zero social welfare and revenue.

	   % Assume for some set of signals $\s_{-i}$ we have that $s_{-i}^*\leq \t(\s_{-i}) < 1$ (otherwise, the mechanism never allocates). Then we have 
	   % \begin{eqnarray*}
	   %     & &\E_{\tilde{\s}_{-i}\sim U[0,1]^{n-1}}\left[v_i(t_i(s_{-i}),\tilde{\s}_{-i})\right] \\&=& \E_{\tilde{\s}_{-i}\sim U[0,1]^n}\left[\max\{t_i(s_{-i}),\tilde{\s}_{-i}\}\right]  \\
	   %     & = & t(\s_{-i})\Pr_{\tilde{\s}_{-i}\sim U[0,1]^n }\left[\max\{\tilde{\s}_{-i}\}\leq t(\s_{-i})\right] +\\
	   %     & & \E_{\tilde{\s}_{-i}\sim U[0,1]^{n-1}}[\max\{\tilde{\s}_{-i}\}\ |\ \max\{\tilde{\s}_{-i}\} > t(\s_{-i})]\Pr_{\tilde{\s}_{-i}\sim U[0,1]^{n-1}}\left[\max\{\tilde{\s}_{-i}\}> t(\s_{-i})\right] \\
	   %     & > & t(\s_{-i})\ =\  v_i(t_i(\s_{-i}), \s_{-i}),
	   % \end{eqnarray*}
	   % The last equality follows since $t(\s_{-i})\geq s_i^*$. 
	   % The inequality follows since $\E_{\tilde{\s}_{-i}\sim U[0,1]^{n-1}}[\max\{\tilde{\s}_{-i}\}\ |\ \max\{\tilde{\s}_{-i}\} > t(\s_{-i})] > t(\s_{-i})$ and since  $$\Pr_{\tilde{\s}_{-i}\sim U[0,1]^{n-1} }\left[\max\{\tilde{\s}_{-i}\}> t(\s_{-i})\right] = 1-(1-t(\s_{-i}))^{n-1} > 0.$$ 
	   % Therefore, 
	   % \begin{eqnarray*}
	   %     v_i(t(\s_{-i}), \s_{-i}) & < & (1-\chi)v_i(t(\s_{-i}), \s_{-i}) + \chi  \E_{\tilde{\s}_{-i}\sim U[0,1]^n}\left[v_i(t_i(s_{-i}),\tilde{\s}_{-i})\right]\\
	   %     & = & v_i^\chi(t(\s_{-i}), \s_{-i}),
	   % \end{eqnarray*}
	   % and $p_i(0,\s_{-i})=\min\{0, v_i(t(\s_{-i}) - v_i^\chi(t(\s_{-i}), \s_{-i})\} < 0$. Since the $\max$ function satisfies the cursedness-monotonicity, Theorem~\ref{thm:no-positive-transfers} implies $t$ cannot be supported by and $\epbb$ mechanism.
	\end{proof}

\subsection{Proof of Lemma~\ref{lem_recursive_expectation}}
\begin{proof}
\begin{align*}\E_{\z}\left[q(\z)|\sum z_j \ge d\right]
= & \E_{\z_{-i}}\left[\E_{z_i|\z_{-i}}\left[q(\z)|z_i\ge d-\sum_{j\ne i}z_j\right]\big| \sum_{j\ne i} z_j \ge d-b\right] \\
\ge & \E_{\z_{-i}}\left[\E_{z_i|\s_{-i}}\left[q(\z)\right]|\sum_{j\ne i} z_{j}\ge d-b\right]
\end{align*}
The equality is because when the supremum of the support of any $z_i$ is $b$, then $\sum z_j\ge d$ if and only if $\sum_{j\ne i} z_j\ge d-b$ and $z_i\ge c-\sum_{j\ne i}z_j$.
The inequality is because $v(\z)$ is non-decreasing in $z_i$ given any $\z_{-i}$.
\end{proof}

% \subsection{Proof of Lemma~\ref{lem_recursive_FOSD}}
% \begin{proof}
% This is due to the definition of FOSD.
% \end{proof}

\section{Derivation of Equation~\eqref{eq_EPEU_ICEU}}
\label{sec:bic-epic-der}

%\subsubsection{Derivation of Eq~(\ref{eq_EPEU_ICEU})}
    \begin{align*}
 	& \int_{\b_{-i}}f_\sigma(\b_{-i}|s_i)EU_i^\chi(\b,s_i;\sigma_{-i})d\b_{-i} \\
 	= &
 	\int_{\b_{-i}}\int_{\s_{-i}}f_\sigma(\b_{-i}|s_i)
	\bigg(
	(1-\chi)f_\sigma(\s_{-i}|\b_{-i},s_i)u_i(\b,\s)+\chi f(\s_{-i}|s_i)u_i(\b,\s)\bigg)d\s_{-i}d\b_{-i}\\
	= & 
	\int_{\b_{-i}}\int_{\s_{-i}}
	\bigg(
	(1-\chi)f_\sigma(\s_{-i},\b_{-i}|s_i)u_i(\b,\s)+\chi \tilde{f}(\b_{-i},\s_{-i}|s_i)u_i(\b,\s)\bigg)d\s_{-i}d\b_{-i}\\
	= & \int_{\b_{-i}}\int_{\s_{-i}}
	f^\chi_\sigma(\b_{-i},\s_{-i}|s_i)u_i(\b,\s)d\s_{-i}d\b_{-i} = EU^\chi_i(b_i,s_i;\sigma_{-i})
    \end{align*}
	
\end{document}